\documentclass{article}

\usepackage{arxiv}

\usepackage[utf8]{inputenc}
\usepackage[T1]{fontenc}
\usepackage{hyperref}
\usepackage{url}
\usepackage{booktabs}
\usepackage{amsfonts}
\usepackage{microtype}
\usepackage{graphicx}
\usepackage{natbib}
\usepackage{doi}

\usepackage{amsmath}
\usepackage{amsthm}
\usepackage{mathtools}
\usepackage{dsfont}
\usepackage{array}
\usepackage{subcaption}
\usepackage{float}
\usepackage{longtable}
\usepackage{algorithm}
\usepackage{algpseudocode}
\usepackage{appendix}
\usepackage{adjustbox}
\usepackage{xcolor}
\usepackage{hhline}
\usepackage{rotating}
\usepackage{wrapfig}
\usepackage{sidecap}

\allowdisplaybreaks[2]

\theoremstyle{definition}
\newtheorem{definition}{Definition}[section]
\newtheorem{theorem}{Theorem}[section]
\newtheorem{corollary}{Corollary}[theorem]
\newtheorem{lemma}[theorem]{Lemma}
\newtheorem{Proposition}{Proposition}[section]

\theoremstyle{remark}

\title{Renewing Reliability: Valuation and Credit Risk Adjustments for Renewable Power Purchase Agreements}

\author{
Nicola Bartolini\thanks{Corresponding author: nicola.bartolini11@unibo.it}\\
Department of Statistics\\
University of Bologna\\
Bologna, Italy\\
\texttt{nicola.bartolini11@unibo.it}
\And
Silvia Romagnoli\\
Department of Statistics\\
University of Bologna\\
Bologna, Italy
\And
Amia Santini\\
Department of Statistics\\
University of Bologna\\
Bologna, Italy
}

\hypersetup{
pdftitle={Renewing Reliability: Valuation and Credit Risk Adjustments for Renewable Power Purchase Agreements},
pdfauthor={Nicola Bartolini, Silvia Romagnoli, Amia Santini},
pdfkeywords={Power Purchase Agreements, Renewable Energy, Risk Management, Monte Carlo Simulation, Electricity Markets}
}

\begin{document}
\maketitle

\begin{abstract}
	Power Purchase Agreements (PPAs) are bilateral over-the-counter contracts central to renewable energy financing. While their capacity to stabilise revenues and hedge price risk is well recognised, their OTC structure exposes both parties to counterparty credit risk. This is a dimension yet to be explored in the literature, particularly given the dual price and volumetric uncertainty inherent in renewable sources. This paper develops a framework for the pricing and valuation of wind power PPAs and for quantifying this risk through Credit Valuation Adjustment (CVA) and Debit Valuation Adjustment (DVA). We model the joint dynamics of electricity spot prices and renewable output, incorporate default probabilities, and compute valuation adjustments that reflect the fair value of bilateral credit risk. The framework provides market participants with a transparent metric for PPA valuation under counterparty risk. While initiatives such as the European Investment Bank's pilot guarantee scheme aim to mitigate credit risk for certain offtakers, such interventions do not cover all PPA transactions. Rigorous internal credit risk assessment therefore remains indispensable for lenders, producers, and offtakers alike.
\end{abstract}


\section{Introduction}
A Power Purchase Agreement (PPA) constitutes a contractual arrangement between an electricity generator and a consumer, referred to respectively as the producer and offtaker. Under this arrangement, the offtaker commits to acquiring specified volumes of electricity from the producer according to predetermined schedules and pricing structures. These agreements exhibit considerable variation across several key dimensions, including the nature of the participating entities (whether private sector or public institutions), the classification of the energy source (renewable versus conventional), and the chosen settlement methodology (physical delivery or financial instruments, \cite{Hundt202135}). The growing scholarly and industrial attention devoted to PPAs stems from their capacity to mitigate the inherent uncertainties associated with renewable energy generation, thereby potentially accelerating the deployment and advancement of these technologies (\cite{POMBOROMERO2024107861}, \cite{baringa2022commercial}).
\\
Our investigation centers specifically on commercial PPAs, which differ from government-backed arrangements and involve renewable energy sources. These agreements may be structured either as physical or financial (commonly termed virtual) contracts. Under physical arrangements, electricity flows directly from the production facility to the purchaser via the transmission network, with the consumer receiving regular invoices reflecting the contractually stipulated price. Conversely, financial PPAs operate through a different mechanism: the renewable generator participates in wholesale electricity markets, selling output at prevailing spot prices. Simultaneously, a financial settlement occurs based on the differential between the agreed PPA price and the market spot price. When spot prices fall beneath the contractual threshold, the buyer compensates the generator; the payment direction reverses when market prices exceed the PPA price.
\\
A fundamental characteristic of PPAs, particularly those of a commercial nature, is that they are typically negotiated and traded as bilateral over-the-counter (OTC) contracts rather than on organised exchanges. As highlighted in analyses of electricity trading, this OTC structure inherently exposes both parties to counterparty credit risk, i.e. the possibility that one party may default on its obligations before the contract matures. While exchange-traded contracts benefit from the risk mitigation provided by a central counterparty, OTC agreements like PPAs require direct and sophisticated management of this bilateral risk.
\\
In the aftermath of the global financial crisis, significant advances have been made in the quantitative modelling of such bilateral risks. The foundational work of \cite{brigo2013counterparty} provides a comprehensive analytical framework for the dynamic valuation and hedging of bilateral counterparty risk on OTC derivative contracts under funding constraints. Within this framework, two key valuation adjustments have emerged as essential components of fair valuation: the Credit Valuation Adjustment (CVA), which quantifies the expected loss arising from the counterparty's potential default, and the Debit Valuation Adjustment (DVA), which reflects the gain from the entity's own potential default (\cite{green2015xva}). As Turnbull notes in his comprehensive review, estimating these quantities requires careful modelling of default probabilities, loss given default, and the dependence structure among these inputs. However, the application of these sophisticated credit risk metrics to long-term renewable energy contracts remains an underexplored area, particularly given the unique dual risks (price and volume) that characterise these agreements (\cite{PENA2024101513}).
\\
Recognising the barrier that creditworthiness poses to PPA adoption, particularly for smaller corporate offtakers, the European Investment Bank (EIB) has recently launched a pilot guarantee scheme as part of its contribution to the Clean Industrial Deal. Announced in early 2025, this initiative provides counter-guarantees for corporate PPAs with an indicative budget of EUR 500 million\footnote{European Investment Bank, ``PAN-EU POWER PURCHASE AGREEMENT GUARANTEE", Project No. 20250202, approved 19 June 2025. \url{https://www.eib.org/en/projects/all/20250202}}. The scheme operates by having the EIB back a guarantee provided by a third-party financial intermediary, thereby absorbing part of the credit risk that has historically discouraged lenders and energy suppliers from engaging with non-investment grade stakeholders. By supporting midcaps and larger corporates in signing guarantees on corporate PPAs, the EIB aims to help them access green energy supply at long-term fixed prices, secure demand for renewable energy producers, and meet lenders' requirements.
\\
Despite the introduction of such guarantee schemes, the quantitative assessment of counterparty risk in PPAs (and consequently, the evaluation of CVA and DVA) remains critically important for several reasons. Firstly, guarantee schemes do not eliminate risk but rather transfer it, which means that understanding the underlying exposure is essential for pricing such guarantees correctly and for managing the residual risk retained by commercial banks and energy suppliers. Second, the valuation of CVA and DVA provides market participants with a transparent, risk-adjusted metric for comparing the economic trade-offs between OTC PPAs and other hedging instruments. Third, as the PPA market matures and expands to include a wider array of participants, including SMEs through potential multi-buyer structures, the ability to quantify bilateral credit risk will become indispensable for negotiating contract terms, determining collateral requirements, and ensuring the long-term stability of renewable energy investments (\cite{POMBOROMERO2024107861}). 
\\
The present work addresses this gap by developing a framework for the pricing and valuation of different types of commercial renewable energy PPAs and for assessing their counterparty credit risk, adapting tools that remain essential even as public policy initiatives evolve to support market development. The paper is structured as follows: Section \ref{sec:literature} presents a literature review about PPAs and highlights the contributions of the present work, Section \ref{sec:general_framework} introduces the general framework, Section \ref{sec:counterparty_credit_risk} illustrates and develops the computation of Counterparty Credit Risk, Section \ref{sec:empirical} provides an empirical application, and Section \ref{sec:policy} goes over the use of this framework within the current EIB policies. Finally, Section \ref{sec:conclusions} concludes.

\section{Literature Review}\label{sec:literature}
The literature on renewable Power Purchase Agreements (PPAs) has developed along multiple paths. \cite{TAHERI2025256} focus on their role for power procurement, comparing portfolios hat sign physical PPAs with ones that sign virtual PPAs. They restrict their study to fixed-volume contracts and emphasize the uncertainty coming from power demand and electricity prices. They tackle the procurement problem by formulating a Markov decision process with multi-stage optimization that reaches and sustains a renewable procurement target.

\cite{MITTLER2025115293} compile a comprehensive overview of existing renewable PPA types. They categorize them by their main features and expand their literature review by including first-hand expert knowledge obtained from interviews with industry players. We rely on their categorization to select the baseline framework that lies at the basis of our study.

The pricing problem of commercial PPAs has been approached from two perspectives: the first, presented in \cite{MENDICINO2019113577}, relies on a Levelized Cost of Energy formula (LCOE), a tool from engineering economics, which is used to compare the lifetime costs of capital-intensive projects. The authors use it to develop a methodology to define optimal contract length and price, finding the former to be betweenn 7 and 10 years and the latter between 75€/MWh and 100€/MWh.
The second approach to pricing is introduced in \cite{PENA2026125168} and is of a financial nature. The authors perform risk-neutral pricing for the PPA, by setting it equal to the fixed payment that would make the discounted risk-neutral expectation of the payoff equal to zero. They then allow for differences in the actual PPA price, due to the influence of renewable energy credits, guarantees of origin, capacity payments, or periodic fees for operation and maintenance costs. They show that financial-based pricing approaches outperform the cost-based ones on historical data. In our work, we extend the financial-based pricing approach to tackle a novel side of the pricing problem, i.e. accounting for the counterparty risk intrinsic to PPA contracts. We offer a computational framework for finding Credit and Debit Valuation Adjustments for wind power PPAs, which can be used for pricing at contract inception, for risk management, during the contract lifetime, and for policymaking, when defining the terms of guarantee schemes.

The next strand of literature focuses on the risks associated with PPAs. \cite{PENA2024101513} focus on the hedging perspective: using monthly and yearly exchange-traded electricity futures, they find an overall low or negative hedging effectiveness and inconsistent tail risk reductions. This indicates that traditional linear hedging strategies could be insufficient for risk mitigation of the the joint price and production risks of PPAs. 
\cite{POMBOROMERO2024107861} focus instead on credit risk and propose a RE PPA assessment model which involves the main drivers of value and risk for the offtaker: cost and volatility reductions, compared to the electricity market. 
Credit risk is then defined via a structural model and as a function of the probability of default of the offtaker, which is in turn identified in the probability that the value of the PPA becomes negative for them. The main driver of value for the offtaker is the the difference between the cost of energy in the market and via the PPA.
In the present study, we take a financial approach and assume, as in \cite{PENA2026125168}, that the contract value is zero at inception due to the appropriate choice of the fixed PPA price and then (similarly to a swap agreement) changes in sign overtime. We define counterparty credit risk within a CVA and DVA framework, in which default probabilities are intensity-based and are not functions of the sign of the expected contract payoff.

\section{The General Framework}
\label{sec:general_framework} 
\label{sec:pricing}
We consider an offtaker entering a renewable energy PPA at time $t$ (contract execution date) with a producer. The contract follows the standard pay-as-produced structure for renewable energy: the offtaker agrees to purchase all (or a fixed percentage) of the electricity generated by the producer at a fixed price per MWh. In some cases, the price may be time-dependent, reflecting factors such as escalating operations and maintenance costs due to equipment aging \cite{PENA2024101513}. This arrangement introduces volumetric risk, as the quantity to be delivered is unknown at contract signature.

Settlement begins at time $T_1$ (commercial operating date) and ends at time $T=T_n$, with $n$ payment periods. On each period, the offtaker purchases the realised output at a fixed agreed-upon price and may sell it on the market at the prevailing spot price. The offtaker therefore faces both volumetric and price risk: neither the quantity nor the market price is known in advance.
\\
Existing PPA types are very heterogeneous, with differences in contract structure, delivery profile, pricing indexation, settlement frequency, tenor, and credit support mechanisms (\cite{wbcsd_pricing_2021_alt}). In order to proceed with this work, a contract definition is required. We begin by focusing on a contract which has a constant price $K$.

\begin{definition}
    \label{def:PPA_math}
    \textbf{Payoff of a PPA} \newline
    The payoff to the offtaker of a PPA contract signed at time $t$ with maturity $T=T_n$, delivery and settlement periods $T_j$, and a fixed price $K$ is defined as 
    \begin{equation*}
        \Pi(T_1, T) = \sum_{j=1}^n \Pi(T_j) =  \sum_{j=1}^n Q(T_j)\big(S(T_j) - K\big)
    \end{equation*}
    where $\Pi(T_j)$ is the i-th payoff at time $T_j$, $Q(T_j)$ is the quantity of electricity produced at $T_j$ and $S(T_j)$ is the price of electricity at $T_j$.
\end{definition}

\begin{definition}
    \label{def:PPA_discount}
     \textbf{Discounted payoff of a PPA} \newline
    The discounted payoff to the offtaker of a PPA contract at time $t$, with fixed price $K$, maturity $T$ and settlement and delivery periods $T_j$ is 
    \begin{equation}
        \label{eq:PPA_discount}
        \Pi(t, T_1,T) = \sum_{i=1}^n \Pi(T_j)D(t,T_j) =  \sum_{i=1}^n Q(T_j)\big(S(T_j) - K\big)D(t,T_j)
    \end{equation}
    where $D(t, T_j)$ is the discount factor from time $T_j$ to time $t$.
\end{definition}

\begin{definition} 
\label{def:UPPA}
\textbf{The Unitary Power Purchase Agreement (UPPA)}\newline Let $(\Omega, \mathcal{F}, (\mathcal{F}_t)_{t\ge 0}, \mathbb{P})$ be a complete filtered probability space. We define the value at time $t$ of a Unitary Power Purchase Agreement (UPPA) for maturity $T$ as 
\begin{equation}
    \label{eq:UPPA}
    UPPA(t, T) = \mathbb{E}^\mathbb{Q}[\Pi(T)D(t,T)|\mathcal{F}_t] = \mathbb{E}^\mathbb{Q}[Q(T)\big(S(T) - K\big)D(t,T)|\mathcal{F}_t],
\end{equation}
where $\mathbb{Q}$ is the risk-neutral measure equivalent to $\mathbb{P}$.
\end{definition}
By using definition \ref{def:UPPA} inside Eq. \eqref{eq:PPA_discount}, we then have that
\begin{equation*}
    \mathbb{E}^\mathbb{Q}[\Pi(t,T_1,T)|\mathcal{F}_t] = \sum_{j=1}^n UPPA(t, T_j) = \sum_{j=1}^n \mathbb{E}^\mathbb{Q}[Q(T_j)\big(S(T_j) - K\big)D(t,T_j)|\mathcal{F}_t]
\end{equation*}
As a consequence, the discounted expected payoff at time $t$ of the PPA is the sum of the individual UPPA contracts. In the next sections, the UPPAs will act as tool to define the valuation of a PPA from a financial mathematics perspective aligned with the economic interpretation of discounted payoffs.


\subsection{Risk-neutral valuation and pricing}\label{sec:risk_neutral}
We begin by introducing a definition for the risk-neutral fixed price $K$ of a PPA contract, conceptually aligned with \cite{PENA2026125168}.
\begin{definition}\label{def:price}
\textbf{Risk-neutral price of a PPA}\newline
The risk-neutral price of a PPA contract is the constant payment $K \in \mathbb{R}$ that makes the discounted risk-neutral conditional  expectation of the PPA payoff equal to zero, i.e.
\begin{equation*}
    K = \frac{\sum_{j=1}^n \mathbb{E}^\mathbb{Q}[Q(T_j)S(T_j) D(t,T_j)|\mathcal{F}_t]}{\sum_{j=1}^n \mathbb{E}^\mathbb{Q}[Q(T_j) D(t,T_j)|\mathcal{F}_t]}.
\end{equation*}
\end{definition}
Given the absence of liquid financial derivatives contracts on $Q(T)$ or on the risk factors associated with it, our framework, as \cite{PENA2026125168}, assumes $\mathbb{Q} = \mathbb{P}$ for the produced renewable energy quantity.

We next focus on PPA contracts for wind power production. In such a setting, the amount of electricity produced in any period $T_j$, $Q(T_j)$, is a function of $W(T_j)$, the average wind speed for that period. We therefore denote it as $Q(T_j, W(T_j))$.

Following \cite{risks6020056}, we assume that the link between wind speed and power production is captured by 
\begin{equation}\label{eq:wind_power_prod}
Q(T,W(T)) = \alpha W(T)^3 \mathbf{1}_{\{W(T) \in [a,b]\}},
\end{equation}
where $\alpha>0$ is an efficiency production factor, $a$ is the lowest wind speed to achieve electricity production, and $b$ is the highest.
The definition of the risk-neutral price of a wind power PPA can then be updated. Differently from \cite{PENA2026125168}, in our work quantity is an explicit function of wind speed $W(t)$. As a consequence, by taking $\mathbb{Q} = \mathbb{P}$, we are assuming a zero market price of risk for wind speed and therefore also for its role in impacting electricity prices.

\begin{definition}\label{def:price_wind}
\textbf{Risk-neutral price of a wind power PPA}\newline
The risk-neutral price of a wind power PPA contract is the constant payment $K \in \mathbb{R}$ which makes the discounted risk-neutral expectation of the PPA payoff equal to zero, i.e.
\begin{equation}
    \label{eq:price_k_of_wind_ppa}
    K = \frac{\sum_{j=1}^n \mathbb{E}^\mathbb{Q}[\alpha W(T_j)^3 \mathbf{1}_{\{W(T_j) \in [a,b]\}}S(T_j) D(t,T_j)|\mathcal{F}_t]}{\sum_{j=1}^n \mathbb{E}^\mathbb{Q}[\alpha W(T_j)^3 \mathbf{1}_{\{W(T_j) \in [a,b]\}}D(t,T_j)|\mathcal{F}_t]}.
\end{equation}
\end{definition}

The discounted expected payoff then also becomes a function of time, of wind speed $W(t)$ and of $S(t)$. 
\begin{definition}
     \textbf{Discounted expected payoff of a wind power PPA} \newline
    The discounted expected payoff, from the perspective of the offtaker, of a wind power PPA signed at time $t$, with fixed price $K$, maturity $T$ and settlement and delivery periods $T_j$ is 
    \begin{align}
    \nonumber
    \mathbb{E}^\mathbb{Q}[\Pi(t,T,W(t),S(t))|\mathcal{F}_t] & = \sum_{j=1}^n UPPA(t, T_j, W(t), S(t)) 
    \\
    &=  \sum_{j=1}^n \mathbb{E}^\mathbb{Q}[Q(T_j, W(T_j))\big(S(T_j) - K\big)D(t,T_j)|\mathcal{F}_t] \notag \\
    &=  \sum_{j=1}^n \alpha \mathbb{E}^\mathbb{Q}[ W(T_j)^3 \mathbf{1}_{\{W(T_j) \in [a,b]\}}S(T_j)]D(t,T_j)|\mathcal{F}_t] \notag \\ 
    &\quad - \alpha K \mathbb{E}^\mathbb{Q}[ W(T_j)^3 \mathbf{1}_{\{W(T_j) \in [a,b]\}} D(t,T_j)|\mathcal{F}_t]. \label{eq:pricing_wind_ppa_definition}
\end{align}
\end{definition}

A risk-neutral approach is however not entirely appropriate for this setting. In fact, while exchange-traded contracts benefit from the risk mitigation provided by a central counterparty, PPAs are mainly OTC agreements. As a consequence, counterparty credit risk must be explicitly accounted for. We therefore begin by finding formulas for risk-neutral pricing and payoff valuation (Sections \ref{sec:gaussian_model} and \ref{sec:jump_model}) and then move to the Credit Valuation Adjustment (CVA) and the Debit Valuation Adjustment (DVA) of the wind PPA contract (Section \ref{sec:counterparty_credit_risk}). Then, Section \ref{sec:K_with_BVA} provides a definition of the counterparty-risk-adjusted price of a wind power PPA.

In order find solutions for Eq. \eqref{eq:price_k_of_wind_ppa} and \eqref{eq:pricing_wind_ppa_definition}, we first need to introduce a model for the two underlying stochastic processes, $W(t)$ and $S(t)$. We consider two alternatives: initially, a more tractable Gaussian model for both, which is illustrated in Section \ref{sec:gaussian_model} and leads to closed-form formulas. Next, a more realistic model with jump-diffusion Ornstein-Uhlenbeck processes, which is illustrated in Section \ref{sec:jump_model} and leads to semi-analytical formulas.

\subsubsection{PPA Pricing and payoff valuation with the Gaussian model}\label{sec:gaussian_model}
The first approach assumes the following model for the wind speed $W(t)$ 
\begin{equation*} \label{eq:gaussian_wind_process}
    W(t) = \exp\bigg\{ \Lambda_W(t) + Y(t) \bigg\}, 
\end{equation*}
with
\begin{equation*}    \label{eq:gaussian_wind_underlying_process_Y}
    dY(t) = \kappa_W[\theta_W - Y(t)]dt + \sigma_W dZ^W(t),
\end{equation*}
where $\Lambda_W(t)$ is a deterministic function representing the seasonality component and $\kappa_W$, $\theta_W$, $\sigma_W \in \mathbb{R}$ are, respectively, the speed of mean reversion, the long-run mean, and the volatility of the logarithm of the de-seasonalized wind speed $Y(t)$. Additionally, $Z^W(t)$ is a standard Brownian motion.
This implies, by Itô's lemma, that
\begin{equation*}
    \label{eq:gaussian_implicit_SDE_price_wind}
    dW(t) = W(t)\left[\Lambda_W'(t) + \kappa_W(\theta_W - Y(t)) + \frac{1}{2}\sigma_W^2\right]dt + W(t) \sigma_W dZ^W(t),
\end{equation*}
with
\begin{align}
W(T) &= \exp\left\{\Lambda_W(T)  + Y(t)e^{-\kappa_W (T-t)} + \theta_W(1-e^{-\kappa_W (T-t)}) +\notag \right.\\
& \quad + \left. \sigma_W \int_t^T e^{-\kappa_W(T-s)} dZ^W(s)\right\}.\label{eq:wind_sol_gauss}
\end{align}
Finally, we have that 
\begin{equation}\label{eq:lambda_W}
    \Lambda_W(t) = \mu_W + a_W \cos(2\pi t/365) + b_W\sin(2\pi t /365),
\end{equation}
with $\mu_W, a_W, b_W \in \mathbb{R}$.

As for spot the price of electricity, a similar model is assumed, i.e.
\begin{equation*}
    \label{eq:gaussian_electricity_process}
    S(t) = \exp\bigg\{ \Lambda_S(t) + X(t) \bigg\}, 
\end{equation*}
with
\begin{equation*}\label{eq:gaussian_electricity_underlying_process_Y}
    dX(t) = \kappa_S(\theta_S - X(t))dt + \sigma_S dZ^S(t), 
\end{equation*}
which implies, by Itô's lemma, that
\begin{equation*}\label{eq:gaussian_implicit_SDE_price_electricity}
    dS(t) = S(t)\left[\Lambda_S'(t) + \kappa_S(\theta_S - X(t)) + \frac{1}{2}\sigma_S^2\right]dt + S(t) \sigma_S dZ^S(t),
\end{equation*}
with
\begin{align}
S(T) &= \exp\left\{\Lambda_S(T) + X(t)e^{-\kappa_S (T-t)} + \theta_S(1-e^{-\kappa_S (T-t)}) + \right. \notag \\
& \quad \left. + \sigma_S \int_t^T e^{-\kappa_S(T-u)} dZ^S(u)\right\}.\label{eq:electricity_sol_gauss}
\end{align}
We also have that 
\begin{equation}\label{eq:lambda_S}
    \Lambda_S(t) = \mu_S + a_S\cos(2\pi t/365) + b_S\sin(2\pi t /365),
\end{equation}
with $\mu_S, a_S, b_S \in \mathbb{R}$.
Finally, we assume that 
\begin{equation*}
    \mathbb{E}^\mathbb{Q}[dZ^S(t) dZ^W(t)] = \rho dt.
\end{equation*}

\begin{Proposition}\label{prop:val_gauss}
The discounted expected payoff of a wind power PPA (Eq. \eqref{eq:pricing_wind_ppa_definition}), under the Gaussian model, is
\begin{align}
    \nonumber
    \mathbb{E}_t^\mathbb{Q}[\Pi(t,T,W(t),S(t))] & = \sum_{j=1}^n UPPA(t, T_j, W(t), S(t)) \notag
    \\
    &=\alpha \sum_{j=1}^n  D(t,T_j)\Big\{ e^{3\Lambda_W(T_j) + 3Y(t)e^{-\kappa_W(T_j-t)} + 3\theta_W(1-e^{-\kappa_W(T_j-t)})}  \notag \\
    &\quad \times e^{\Lambda_S(T_j) + X(t)e^{-\kappa_S(T_j-t)}  + \theta_S(1-e^{-\kappa_S(T_j-t)}) } \notag\\
    & \quad \times e^{\frac{\sigma_S^2(1-\rho^2)}{4\kappa_S}(1-e^{-2\kappa_S(T_j-t)})  + \frac{9 \sigma_W^2}{4 \kappa_W}(1-e^{-2\kappa_W (T-t)})} \notag  \\
    & \quad \times e^{\frac{\rho^2 \sigma_S^2}{4\kappa_S}(1-e^{-2\kappa_S(T_j-t)})+ \frac{3\rho\sigma_W\sigma_S}{\kappa_W+\kappa_S}(1-e^{-(\kappa_W+\kappa_S)(T_j-t)})} \notag \\
    & \quad \times \big( \Phi(\bar{b})- \Phi(\bar{a}) \big) - K e^{3\Lambda_W(T_j) + 3Y(t)e^{-\kappa_W(T_j-t)}}  \notag \\
    &\quad  \times e^{3\theta_W(1-e^{-\kappa_W(T_j-t)}) + \frac{9\sigma_W^2}{4\kappa_W} (1-e^{-2\kappa_W(T-t)})} \big( \Phi(\tilde{b}) - \Phi(\tilde{a})\big)\Big\} ,\label{eq:final_ppa_gaussian}
\end{align}
where
\begin{align*}
&\tilde{a} = \frac{\ln(a) -  \Lambda_W(T_j) - Y(t)e^{-\kappa_W(T_j-t)} - \theta_W(1-e^{-\kappa_W(T_j-t)}) -\frac{3\sigma_W^2}{2\kappa_W}(1-e^{-2\kappa_W(T_j-t)})}{\sigma_W\sqrt{\frac{1-e^{-2\kappa_W(T_j-t)}}{2\kappa_W}}} , \\
&\tilde{b} = \frac{\ln(b) -  \Lambda_W(T_j) - Y(t)e^{-\kappa_W(T_j-t)} - \theta_W(1-e^{-\kappa_W(T_j-t)}) -\frac{3\sigma_W^2}{2\kappa_W}(1-e^{-2\kappa_W(T_j-t)})}{\sigma_W\sqrt{\frac{1-e^{-2\kappa_W(T_j-t)}}{2\kappa_W}}} , \\
&\bar{a} = \bigg(\ln(a) - \Lambda_W(T_j)  - Y(t)e^{-\kappa_W (T_j-t)} - \theta_W(1-e^{-\kappa_W (T_j-t)}) +\\ 
&\quad - \frac{3\sigma^2_W}{2\kappa_W} (1-e^{-2\kappa_W(T_j-t)}) -\frac{\rho \sigma_S \sigma_W}{\kappa_W + \kappa_S} (1-e^{-(\kappa_W+\kappa_S)(T_j-t)}) \bigg)
\bigg/ \\
& \quad {\sigma_W  \sqrt{\frac{1}{2\kappa_W}(1- e^{-2\kappa_W(T_j-t)})}} \\
&\bar{b} = \bigg(\ln(b) - \Lambda_W(T_j)  - Y(t)e^{-\kappa_W (T_j-t)} - \theta_W(1-e^{-\kappa_W (T_j-t)}) + \\
&\quad - \frac{3\sigma^2_W}{2\kappa_W} (1-e^{-2\kappa_W(T_j-t)}) -\frac{\rho \sigma_S \sigma_W}{\kappa_W + \kappa_S} (1-e^{-(\kappa_W+\kappa_S)(T_j-t)}) \bigg)
\bigg/ \\
& \quad {\sigma_W  \sqrt{\frac{1}{2\kappa_W}(1- e^{-2\kappa_W(T_j-t)})}}.
\end{align*}
\end{Proposition}
\begin{proof}
    The proof is provided in \ref{app:ppa_gaussian_val}, under the assumption of deterministic interest rates.
\end{proof}

\begin{Proposition}\label{prop:K_gaussian}
The risk-neutral constant price $K$ of a wind power PPA contract (Eq. \eqref{eq:price_k_of_wind_ppa}), under the Gaussian model, is
\begin{align}
    K &= \bigg(\sum_{j=1}^n D(t,T_j) e^{3\Lambda_W(T_j) + 3Y(t)e^{-\kappa_W(T_j-t)} + 3\theta_W(1-e^{-\kappa_W(T_j-t)}) + \Lambda_S(T_j)}  \notag \\
    &\quad \times e^{X(t)e^{-\kappa_S(T_j-t)}  + \theta_S(1-e^{-\kappa_S(T_j-t)}) +\frac{\sigma_S^2(1-\rho^2)}{4\kappa_S}(1-e^{-2\kappa_S(T_j-t)}) + \frac{9 \sigma_W^2}{4 \kappa_W}(1-e^{-2\kappa_W (T-t)})} \notag\\
    & \quad \times e^{\frac{\rho^2 \sigma_S^2}{4\kappa_S}(1-e^{-2\kappa_S(T_j-t)}) + \frac{3\rho\sigma_W\sigma_S}{\kappa_W+\kappa_S}(1-e^{-(\kappa_W+\kappa_S)(T_j-t)})} \big( \Phi(\bar{b})- \Phi(\bar{a}) \big) \bigg) \bigg/ \notag \\
    & \quad  \bigg( \sum_{j=1}^n D(t,T_j) e^{3\Lambda_W(T_j) + 3Y(t)e^{-\kappa_W(T_j-t)} + 3\theta_W(1-e^{-\kappa_W(T_j-t)}) +\frac{9\sigma_W^2}{4\kappa_W} (1-e^{-2\kappa_W(T-t)})} \notag \\
    & \quad \times \big( \Phi(\tilde{b}) - \Phi(\tilde{a})\big)\bigg)\label{eq:K_gauss}
\end{align}
where
\begin{align*}
&\tilde{a} = \frac{\ln(a) -  \Lambda_W(T_j) - Y(t)e^{-\kappa_W(T_j-t)} - \theta_W(1-e^{-\kappa_W(T_j-t)}) -\frac{3\sigma_W^2}{2\kappa_W}(1-e^{-2\kappa_W(T_j-t)})}{\sigma_W\sqrt{\frac{1-e^{-2\kappa_W(T_j-t)}}{2\kappa_W}}} , \\
&\tilde{b} = \frac{\ln(b) -  \Lambda_W(T_j) - Y(t)e^{-\kappa_W(T_j-t)} - \theta_W(1-e^{-\kappa_W(T_j-t)}) -\frac{3\sigma_W^2}{2\kappa_W}(1-e^{-2\kappa_W(T_j-t)})}{\sigma_W\sqrt{\frac{1-e^{-2\kappa_W(T_j-t)}}{2\kappa_W}}}, \\
&\bar{a} = \bigg(\ln(a) - \Lambda_W(T_j)  - Y(t)e^{-\kappa_W (T_j-t)} - \theta_W(1-e^{-\kappa_W (T_j-t)}) - \frac{3\sigma^2_W}{2\kappa_W} (1-e^{-2\kappa_W(T_j-t)})\\ 
&\quad -\frac{\rho \sigma_S \sigma_W}{\kappa_W + \kappa_S} (1-e^{-(\kappa_W+\kappa_S)(T_j-t)}) \bigg)
\bigg/ {\sigma_W  \sqrt{\frac{1}{2\kappa_W}(1- e^{-2\kappa_W(T_j-t)})}} \\
&\bar{b} = \bigg(\ln(b) - \Lambda_W(T_j)  - Y(t)e^{-\kappa_W (T_j-t)} - \theta_W(1-e^{-\kappa_W (T_j-t)}) - \frac{3\sigma^2_W}{2\kappa_W} (1-e^{-2\kappa_W(T_j-t))}) \\
&\quad -\frac{\rho \sigma_S \sigma_W}{\kappa_W + \kappa_S} (1-e^{-(\kappa_W+\kappa_S)(T_j-t)}) \bigg)
\bigg/ {\sigma_W  \sqrt{\frac{1}{2\kappa_W}(1- e^{-2\kappa_W(T_j-t)})}}.
\end{align*}
\end{Proposition}
\begin{proof}
    The proof is provided in \ref{app:ppa_gaussian_price}, under the assumption of deterministic interest rates. 
\end{proof}

\subsubsection{PPA Pricing and payoff valuation with the jump-diffusion Ornstein-Uhlenbeck model}\label{sec:jump_model}
In the second approach we assume an analogous model to Section \ref{sec:gaussian_model} for wind speed $W(t)$, but, following \cite{risks6020056}, $Y(t)$ is a jump-diffusion Ornstein-Uhlenbeck process, i.e. 
\begin{align}
    W(t) &= \exp\bigg\{ \Lambda_W(t) + Y(t) \bigg\} \notag \\ 
    dY(t)&= \kappa_W[\theta_W - Y(t)]dt + \sigma_W dZ^W(t) + dL^W(t), \notag
\end{align}
where $\Lambda_W(t)$ is a deterministic function representing the seasonality component and $\kappa_W$, $\theta_W$, $\sigma_W \in \mathbb{R}$ are, respectively, the speed of mean reversion, the long-run mean, and the volatility of the logarithm of the de-seasonalized wind speed $Y(t)$. Additionally, $Z^W(t)$ is a standard Brownian motion and $L^W(t)$ is a compound Poisson process with Gaussian jumps.
Applying Itô's lemma for jump processes gives
\begin{align}
\label{eq:implicit_SDE_wind_intensity}
dW(t) &= W(t)\left[\Lambda_W'(t) + \kappa_W(\theta_W - Y(t)) + \frac{1}{2}\sigma_W^2\right]dt + W(t) \sigma_W dZ^W(t) \notag \\
& \quad +\int_\mathbb{R} W(t^-) \alpha_W \mathcal{J}^Y(dt, d\alpha_W),
\end{align}
in which the following relationship is exploited
\begin{align*}
    &\int_\mathbb{R} W(t^-) \alpha_W dL^W(t)= \int_0^t\int_\mathbb{R} W(s^-) \alpha_W \mathcal{J}^Y(ds, d\alpha_W),
\end{align*}
where $\mathcal{J}(ds, d\alpha_W)$ is the Poisson Random measure and $\alpha_W$ is a random variable representative of the wind jump size. 
As for the second stochastic component, we follow \cite{risks6020056} and \cite{benth2021MultivariateWind} and model the price of electricity $S(t)$ as 
\begin{equation*}
    S(t) = \exp\bigg\{ \Lambda_S(t) + X(t) \bigg\} 
\end{equation*}
\begin{equation*}
    dX(t) = \kappa_S(\theta_S - X(t))dt + \sigma_S dZ^S(t) + dL^S(t)
\end{equation*}  

\begin{align}
    dS(t) &= S(t)\left[\Lambda_S'(t) + \kappa_S(\theta_S - \log(S(t))) + \frac{1}{2}\sigma_S^2\right]dt + S(t) \sigma_S dZ^S(t) \notag \\
    &\quad +\int_\mathbb{R} S(t^-) \alpha_S \mathcal{J}^S(dt, d\alpha_S), \label{eq:implicit_SDE_price_electricity}
\end{align}
where $\Lambda_S(t)$ is a deterministic function representing the seasonality component.
Moreover, we assume that 
\begin{equation}\label{eq:electricity_wind_dependence_condition}
    \mathbb{E}^\mathbb{Q}[dZ^S(t) dZ^W(t)] = \rho dt.
\end{equation}

\begin{Proposition}\label{prop:val_jump}
The discounted expected payoff of a wind power PPA (Eq. \eqref{eq:pricing_wind_ppa_definition}) under the jump-diffusion Ornstein-Uhlenbeck model is
\begin{align}
\mathbb{E}^\mathbb{Q}_t[\Pi(t,T,W(t),S(t))] & = \sum_{j=1}^n UPPA(t, T_j, W(t), S(t)) 
    \notag\\
    &= \alpha D(t,T_j)\Bigg\{e^{3\Lambda_W(T_j) + 3Y(t)e^{-\kappa_W(T_j-t)} + 3\theta_W(1-e^{-\kappa_W(T_j-t)})}  \notag\\ 
    &\quad \times e^{X(t)e^{-\kappa_S(T_j-t)} + \theta_S(1-e^{-\kappa_S(T_j-t)})} \mathbb{E}^\mathbb{Q}_t[e^{B}]\mathbb{E}^\mathbb{Q}[e^{3(U+V) + A}] \notag \\
    & \quad \times \bigg(\frac{1}{\pi} \int_0^{+\infty} \frac{Im\{e^{-i\xi \tilde{a}} \bar{\varphi}_{U+V}(\xi) \}}{\xi} d\xi \notag \\
    &\quad - \frac{1}{\pi} \int_0^{+\infty} \frac{Im\{e^{-i\xi \tilde{b}} \bar{\varphi}_{U+V}(\xi) \}}{\xi} d\xi\bigg) \notag \\
    & \quad - K \Bigg( e^{3\Lambda_W(T_j) + 3Y(t)e^{-\kappa_W(T_j-t)} + 3\theta_W(1-e^{-\kappa_W(T_j-t)})} \notag \\
    & \quad \times \mathbb{E}^\mathbb{Q}[e^{3(U+V)}] \bigg(\frac{1}{\pi} \int_0^{+\infty} \frac{Im\{e^{-i\xi \tilde{a}} \tilde{\varphi}_{U+V}(\xi) \}}{\xi} d\xi +\notag \\
    &\quad - \frac{1}{\pi} \int_0^{+\infty} \frac{Im\{e^{-i\xi \tilde{b}} \tilde{\varphi}_{U+V}(\xi) \}}{\xi} d\xi\bigg) \Bigg) \Bigg\}, \label{eq:final_ppa_jump}
\end{align}
where 
\begin{align*}
    &\tilde{a} = \log(a) - (\Lambda_W(T_j) + Y(t) e^{-\kappa_W(T_j-t)} + \theta_W (1 - e^{-\kappa_W(T_j-t)} ))
    \\
    & \tilde{b} = \log(b) - (\Lambda_W(T_j) + Y(t) e^{-\kappa_W(T_j-t)} + \theta_W (1 - e^{-\kappa_W(T_j-t)} )) \\
    &\bar{\varphi}_{U+V}(\xi) =  \frac{\varphi_{U,A}(\xi - 3i, -i )\varphi_{V}(\xi - 3i)}{\varphi_{U,A}(- 3i, -i )\varphi_{V}(- 3i)} \\
    &\tilde{\varphi}_{U+V}(\xi) = \frac{\varphi_{U}(\xi - 3i)\varphi_{V}(\xi - 3i)}{\varphi_U(-3i)\varphi_V(-3i)}.
\end{align*}
\end{Proposition}
\begin{proof}
The proof is provided in \ref{app:ppa_jump_val}, under the assumption of deterministic interest rates. 
\end{proof}

\begin{Proposition}\label{prop:K_jump}
The risk-neutral constant price $K$ of a wind power PPA contract (Eq. \eqref{eq:price_k_of_wind_ppa}), under the jump-diffusion Ornstein-Uhlenbeck model, is
\begin{align}
    K 
    &=\bigg( \sum_{j=1}^n D(t, T_j)e^{3\Lambda_W(T_j) + 3Y(t)e^{-\kappa_W(T_j-t)} + 3\theta_W(1-e^{-\kappa_W(T_j-t)}) + X(t)e^{-\kappa_S(T_j-t)}}  \notag\\ 
    &\quad \times e^{\theta_S(1-e^{-\kappa_S(T_j-t)})} \mathbb{E}^\mathbb{Q}_t[e^{B}]\mathbb{E}^\mathbb{Q}[e^{3(U+V) + A}] \big(\frac{1}{\pi} \int_0^{+\infty} \frac{Im\{e^{-i\xi \tilde{a}} \bar{\varphi}_{U+V}(\xi) \}}{\xi} d\xi \notag \\
    &\quad - \frac{1}{\pi} \int_0^{+\infty} \frac{Im\{e^{-i\xi \tilde{b}} \bar{\varphi}_{U+V}(\xi) \}}{\xi} d\xi\big) \bigg) \bigg/ \notag \\
    &\quad \Bigg( \sum_{j=1}^n D(t, T_j) e^{3\Lambda_W(T_j) + 3Y(t)e^{-\kappa_W(T_j-t)} + 3\theta_W(1-e^{-\kappa_W(T_j-t)})} \mathbb{E}^\mathbb{Q}[e^{3(U+V)}] \notag \\
    &\quad \times \bigg(\frac{1}{\pi} \int_0^{+\infty} \frac{Im\{e^{-i\xi \tilde{a}} \tilde{\varphi}_{U+V}(\xi) \}}{\xi} d\xi - \frac{1}{\pi} \int_0^{+\infty} \frac{Im\{e^{-i\xi \tilde{b}} \tilde{\varphi}_{U+V}(\xi) \}}{\xi} d\xi\bigg) \Bigg) \label{eq:K_jump}
\end{align}
where 
\begin{align*}
    &\tilde{a} = \log(a) - (\Lambda_W(T_j) + Y(t) e^{-\kappa_W(T_j-t)} + \theta_W (1 - e^{-\kappa_W(T_j-t)} ))
    \\
    & \tilde{b} = \log(b) - (\Lambda_W(T_j) + Y(t) e^{-\kappa_W(T_j-t)} + \theta_W (1 - e^{-\kappa_W(T_j-t)} )) \\
    &\bar{\varphi}_{U+V}(\xi) =  \frac{\varphi_{U,A}(\xi - 3i, -i )\varphi_{V}(\xi - 3i)}{\varphi_{U,A}(- 3i, -i )\varphi_{V}(- 3i)} \\
    &\tilde{\varphi}_{U+V}(\xi) = \frac{\varphi_{U}(\xi - 3i)\varphi_{V}(\xi - 3i)}{\varphi_U(-3i)\varphi_V(-3i)}.
\end{align*}
\end{Proposition}
\begin{proof}
The proof is provided in \ref{app:ppa_jump_price}, under the assumption of deterministic interest rates. 
\end{proof}

\section{The Valuation Adjustments}\label{sec:counterparty_credit_risk}
We next move to the introduction of the intensity-based credit-risk model, presented in Section \ref{sec:intensity_based_model}, of the Counterparty Valuation Adjustment framework, illustrated in Section \ref{sec:cva_framework}, and to the computation of the Counterparty Valuation Adjustments for the renewable energy PPAs, in Section \ref{sec:the_counterparty_credit_risk_in_ppa}. 

\subsection{The intensity-based credit risk model}
\label{sec:intensity_based_model}
To model the default event we follow the framework proposed by \cite{duffie1999modeling}, that can be also found in \cite{Bielecki_Rutkowski} and \cite{mcneil2015quantitative}, and the empirical application of which can be found in \cite{driessen2005default}. In this paragraph we provide a brief summary of the results that are used in our work. 

Let $(\Omega,\mathcal{F},\mathbb{P})$ be a probability space and let $\mathcal{F}_t = \sigma(\{\psi_t : s\leq t\})$ be the filtration generated by some observed background process, define the random time $\tau > 0$ a.s. on $\mathcal{F}$ and denote by $Y(t) = \mathrm{1}_{\{\tau\leq t \}}$ the associated jump indicator and by $\mathcal{H}_t = \sigma\{\mathrm{1}_{\{\tau\leq s\}}: s\leq t\}$ the filtration generated by 
$Y(t)$, then the general filtration for our purpose is defined as $$\mathcal{G}_t = \mathcal{F}_t \vee \mathcal{H}_t.$$
Then, $\tau$ is a stopping time with respect to $\mathcal{G}_t$ and $\mathcal{H}_t$, but not necessarily to $\mathcal{F}_t$. 
\theoremstyle{definition}\label{def:doubly_stochastic_random_time}
\begin{definition}{\textbf{Doubly stochastic random time (from \cite{Bielecki_Rutkowski})}}
\newline  
A random time $\tau$ is said to be doubly stochastic if there exist a positive $\mathcal{F}_t$ adapted process $\lambda_t$, such that $\Lambda_t = \int_0^t\lambda_s ds$ is strictly increasing and finite for every $t>0$ and such that, for all $t\geq 0$,
\begin{equation}
    \label{eq:survial_prob_stochastic_hazard_rates}
    P(\tau > t | \mathcal{F}_\infty) = e^{-\int_0^t\lambda_s ds}.
\end{equation}
In such a case, $\lambda_t$ is referred to as the $\mathcal{F}_t$-conditional hazard process of $\tau$.
\end{definition}
It is an immediate consequence of Eq \eqref{eq:survial_prob_stochastic_hazard_rates} that $P(\tau > t | \mathcal{F}_\infty)$ is $\mathcal{F}_t$ measurable.

\begin{lemma}\textbf{(From \cite{Bielecki_Rutkowski})}
\label{lemma:base_of_the_key_lemma}
For every $t\leq 0$, the following holds:
$$\mathcal{G}^{*}_t = \{A\in\mathcal{G}_t : \exists B\in\mathcal{F}_t, A\cap\{\tau>t\} = B \cap\{\tau > t\} \}.$$
\end{lemma}
This means that, before the default event, the only known events are the ones related to the background filtration.

\begin{lemma}\textbf{(From \cite{Bielecki_Rutkowski})}
\label{lemma:the_key_lemma}
Let $\tau$ be a random time (not necessarily doubly stochastic) such that $P(\tau > t | \mathcal{F}_t)>0$ for all $t\leq 0$, then for every integrable random variable:
\begin{equation*}
    \mathrm{E}[\mathrm{1}_{\{\tau>t\}}X|\mathcal{G}_t] = \mathrm{1}_{\{\tau>t\}}\frac{\mathrm{E}[\mathrm{1}_{\{\tau>t\}}X|\mathcal{F}_t]}{P(\tau > t | \mathcal{F}_t)}.
\end{equation*}
\end{lemma}

\begin{corollary} (\textbf{From \cite{mcneil2015quantitative}})
\label{corollary:corollary_key_lemma}
Let $T>t$ and assume $\tau$ is doubly stochastic with hazard process $\lambda_t$ if then $\Tilde{X}$ is integrable and $\mathcal{F}_T$ measurable, then:
\begin{equation*}
    \mathrm{E}\{\Tilde{X}\mathrm{1}_{\{\tau>T\}}|\mathcal{G}_t\} = \mathrm{1}_{\{\tau>t\}}\mathrm{E}\bigg[e^{-\int_t^T\lambda_sds}\Tilde{X} \bigg|\mathcal{F}_t\bigg].
\end{equation*}
\end{corollary}

We finally assume, as is standard in the literature, that the hazard rate processes have the following form
\begin{align}
    \label{eq:hazard_rate_buyer}
    & d\lambda_I(t) = k_I(\theta_I - \lambda_I(t))dt + \sigma_I \sqrt{\lambda_I(t)} dW_{t}^{(I)} \\
    \label{eq:hazard_rate_seller}
    & d\lambda_C(t) = k_C(\theta_C - \lambda_C(t))dt + \sigma_C \sqrt{\lambda_C(t)} dW_{t}^{(C)}
\end{align}
and that the two processes are conditionally independent, in order to ensure the tractability of the resulting model.

\subsection{The Counterparty Valuation Adjustment  framework}\label{sec:cva_framework}
We adhere to the standard approach to counterparty credit risk presented in \cite{brigo2013counterparty}. Below follow a short summary of the framework and an application to PPAs.

\begin{definition} \textbf{Credit Valuation Adjustment (CVA)}.\newline
    The Credit Valuation Adjustment is the expected loss arising from the possible default of the counterparty. Formally, if the counterparty defaults at a time when the derivative has positive value to the non-defaulting party, the latter may not recover the full amount owed. The CVA is therefore defined as the risk-neutral expectation of the discounted loss due to counterparty default.
\end{definition}

\begin{definition} \textbf{Debit Valuation Adjustment (DVA)}.\newline
    The Debit Valuation Adjustment is the expected gain arising from the possible default of the reporting entity itself. If the entity defaults at a time when the derivative has negative value to it, it may avoid paying the full amount owed. DVA therefore represents the value of this limited liability: the entity's own default risk reduces the fair value of its liabilities. The DVA is therefore defined as the risk-neutral expectation of the discounted gain arising from the entity's own default.
\end{definition}

Let $\tau_I$ be the default time of the investor (in the standard literature, the bank, and in our setting, the PPA offtaker), let $\tau_C$ be the default time of the PPA energy producer (in the standard literature, the corporation), and let $\tau = min(\tau_I, \tau_C)$.

\begin{definition}\textbf{General bilateral counterparty risk valuation formula (from \cite{brigo2013counterparty})}\newline 
At valuation time $t$ and on the event of no default until $t$, $\{\tau > t\}$, the price of a generic contract with payoff at maturity $T$ under bilateral counterparty risk, $\bar\Pi(t,T)$, is   \begin{align*}
        \mathbb{E}^\mathbb{Q}[\bar\Pi(t,T)|\mathcal{G}_t] &= \mathbb{E}^\mathbb{Q}[\Pi(t,T)|\mathcal{G}_t] + \mathbb{E}^\mathbb{Q}[\mathrm{LGD_I}\mathbf{1}_{\{\mathcal{I}_1 \cup \mathcal{I}_2\}}D(t, \tau_I)\mathbb{E}^\mathbb{Q}[(-\Pi(\tau_I,T))|\mathcal{F}_{\tau_I}]^+|\mathcal{G}_t]  \\
        & \quad
        - \mathbb{E}^\mathbb{Q}[\mathrm{LGD_C}\mathbf{1}_{\{\mathcal{I}_3 \cup \mathcal{I}_4\}}D(t, \tau_C)(\mathbb{E}^\mathbb{Q}[\Pi(\tau_C,T)|\mathcal{F}_{\tau_C}])^+|\mathcal{G}_t],
    \end{align*}
where $D(t, \tau_I)$ is the discount factor from the default time $\tau_I$ to time $t$, $D(t, \tau_C)$ is the discount factor from the default time $\tau_C$ to time $t$, $\mathrm{LGD_I}$ is the loss given default to the creditor in the case of investor default and $\mathrm{LGD_C}$ is the loss given default to the investor in the case of default of the creditor. Additionally, the sets $\mathcal{I}$ have the following definitions:  
 $\mathcal{I}_1 = \{\tau_I < \tau_C < T \}$ 
 $\mathcal{I}_2 = \{\tau_I < T < \tau_C \}$ 
 $\mathcal{I}_3 = \{\tau_C < \tau_I < T \}$ 
 $\mathcal{I}_4 = \{\tau_C < T < \tau_I \}$ 
 $\mathcal{I}_5 = \{T < \tau_I < \tau_C \}$ 
 $\mathcal{I}_6 = \{T < \tau_C < \tau_I \}$. Furthermore, 
 $\mathbb{E}^\mathbb{Q}[\Pi(\tau_I,T)|\mathcal{F}_{\tau_I}]$ is the expected residual value of the contract at the time of investor default, if the investor defaults, while $\mathbb{E}^\mathbb{Q}[\Pi(\tau_C,T)|\mathcal{F}_{\tau_C}]$ is the expected residual value of the contract at the time of counterparty default, if the counterparty defaults. 
We finally remark that in the event $\{\tau>t\}$ we have that $\mathbb{E}^\mathbb{Q}[\Pi(t,T)|\mathcal{G}_t] = \mathbb{E}^\mathbb{Q}[\Pi(t,T)|\mathcal{F}_t]$. 
Explicit formulas for the CVA and DVA are then given by
\begin{align*}
    & \mathrm{CVA}(t,T) = \mathbb{E}^\mathbb{Q}[\mathrm{LGD_C}\mathbf{1}_{\{\mathcal{I}_3 \cup \mathcal{I}_4\}}D(t, \tau_C)\mathbb{E}^\mathbb{Q}[\Pi(\tau_C,T)|\mathcal{F}_{\tau_C}]^+|\mathcal{G}_t], \\
    & \mathrm{DVA}(t,T) = \mathbb{E}^\mathbb{Q}[\mathrm{LGD_I}\mathbf{1}_{\{\mathcal{I}_1 \cup \mathcal{I}_2\}}D(t, \tau_I)(\mathbb{E}^\mathbb{Q}[-\Pi(\tau_I,T)|\mathcal{F}_{\tau_I}]^+|\mathcal{G}_t].
\end{align*}
We finally define the Bilateral Value Adjustment as 
\begin{equation*}
    \mathrm{BVA}(t,T) = \mathrm{DVA}(t,T) - \mathrm{CVA}(t,T).
\end{equation*}
\end{definition}
Then, the PPA valuation formula under counterparty risk becomes
\begin{equation}
    \label{eq:ppa_pricing_formula}
    \begin{split}
        \mathbb{E}^\mathbb{Q}[\bar\Pi(t,T)|\mathcal{F}_t] &= \mathbb{E}^\mathbb{Q}[\Pi(t,T)|\mathcal{F}_t] + \mathrm{BVA}.
    \end{split}
\end{equation}

The calculation of $\mathbb{E}^\mathbb{Q}[\Pi(\tau_C,T)|\mathcal{F}_{\tau_C}]$ and $\mathbb{E}^\mathbb{Q}[\Pi(\tau_I,T)|\mathcal{F}_{\tau_I}]$, required for counterparty risk evaluation, is
approximated by discretizing the timeline $[t,T)$ with a finite series of buckets $[T_i, T_{i+1})$, $i \in \{0,n_b -1\}$, in which the actual default times $\tau_C$ and $\tau_I$ can occur. As in the cited literature, it is also assumed that, in case of default, liquidation happens at the initial time of each bucket. Then, the BVA becomes
\begin{equation}\label{eq:approximated_bva_formula}
    \begin{split}
        \mathrm{BVA}(t,T) &= - \mathrm{LGD_C} \sum_{i=0}^{n_b -1} \mathbb{E}^\mathbb{Q} [ \mathbf{1}_{\{\tau_I > \tau_C\}} \mathbf{1}_{\{T_i \le \tau_C < T_{i+1}\}} D(t,T_i) \left( \mathbb{E}^\mathbb{Q}\left[\Pi(T_i,T)|\mathcal{F}_{T_i}\right]^+ \right)|\mathcal{G}_t] \\
        &
        + \mathrm{LGD_I} \sum_{i=0}^{n_b -1} \mathbb{E}^\mathbb{Q} [ \mathbf{1}_{\{\tau_C > \tau_I\}} \mathbf{1}_{\{T_i \le \tau_I < T_{i+1}\}} D(t,T_i) \left( \mathbb{E}^\mathbb{Q}\left[(-\Pi(T_i,T))|\mathcal{F}_{T_i}\right]^+ \right)|\mathcal{G}_t].
    \end{split}
\end{equation}

We next move Eq. \eqref{eq:approximated_bva_formula} from the general, unobserved filtration $\mathcal{G}_t$, to the observed market filtration $\mathcal{F}_t$. 

\begin{definition}\label{prop:cva_dva_1}
    \textbf{(From \cite{brigo2013counterparty})} \newline
    The CVA and DVA under the market filtration $\mathcal{F}_t$ are
    \begin{align}
\label{eq:CVA_pricing_formula_under_market_filtration}
        \mathrm{CVA}(t,T) &= \mathrm{LGD_C}\sum_{i=0}^{n_b -1}\mathbb{E}^{\mathbb{Q}}\bigg[D(t,T_i)(\mathbb{E}^\mathbb{Q}[\Pi(T_i, T)|\mathcal{F}_{T_i}]^+) \notag \\
        &\quad \times \int_{T_i}^{T_{i+1}}\lambda_C(u)\exp\bigg\{\int_t^u \lambda_C(s)+\lambda_{I}(s) ds \bigg\} \bigg|\mathcal{F}_t \bigg] \\
\label{eq:DVA_pricing_formula_under_market_filtration}
        \mathrm{DVA}(t,T) &= \mathrm{LGD_I}\sum_{i=0}^{n_b -1}\mathbb{E}^{\mathbb{Q}}\bigg[D(t,T_i)(\mathbb{E}^\mathbb{Q}[(-\Pi(T_i, T))|\mathcal{F}_{T_i}]^+) \notag \\
        &\quad \times \int_{T_i}^{T_{i+1}}\lambda_I(u)\exp\bigg\{\int_t^u \lambda_C(s)+\lambda_{I}(s) ds \bigg\} \bigg|\mathcal{F}_t \bigg],
    \end{align}
where $\lambda_C$ and $\lambda_I$ are the hazard processes of the corporation and the investor, respectively.
\end{definition}
The proof can be found in \cite{green2015xva} and \cite{brigo2013counterparty}. 
The approach follows the \textit{wrong-way to risk} method of \cite{brigo2013counterparty}, which assumes the conditional independence of $\lambda_C$ and $\lambda_I$.
The next proposition then follows.
\begin{Proposition}\label{prop:cva_dva_2}\textbf{(From \cite{brigo2013counterparty})}\newline
    The CVA and DVA under the market filtration $\mathcal{F}_t$ are
    \begin{align}
        \label{eq:CVA_pricing_formula_under_market_filtration_v2}
        \mathrm{CVA}(t,T) &\approx \mathrm{LGD_C}\sum_{i=0}^{n_b -1}\mathbb{E}^{\mathbb{Q}}\bigg[D(t,T_i)(\mathbb{E}^\mathbb{Q}[\Pi(T_i, T)|\mathcal{F}_{T_i}]^+) 
        e^{-\int_t^{T_{i+1}} \lambda_I(s)ds} \times  \notag \\
        &\quad \times \bigg(e^{-\int_t^{T_{i}} \lambda_C(s)ds} - e^{-\int_t^{T_{i+1}} \lambda_C(s)ds} \bigg)
        \bigg|\mathcal{F}_t \bigg] \\
        \label{eq:DVA_pricing_formula_under_market_filtration_v2}
        \mathrm{DVA}(t,T) &\approx \mathrm{LGD_I}\sum_{i=0}^{n_b -1}\mathbb{E}^{\mathbb{Q}}\bigg[D(t,T_i)(\mathbb{E}^\mathbb{Q}[(-\Pi(T_i, T))|\mathcal{F}_{T_i}]^+)
        e^{-\int_t^{T_{i+1}} \lambda_C(s)ds}\times \notag \\
        &\quad \times \bigg(e^{-\int_t^{T_{i}} \lambda_I(s)ds} - e^{-\int_t^{T_{i+1}} \lambda_I(s)ds} \bigg) 
        \bigg|\mathcal{F}_t \bigg] 
    \end{align}
\end{Proposition}
The proof again can be found in \cite{green2015xva} and \cite{brigo2013counterparty}.

\subsection{Counterparty Risk in Power Purchase Agreements (PPA)}
\label{sec:the_counterparty_credit_risk_in_ppa}

We now apply this framework to PPAs, maintaining the point of view of the offtaker, as in the payoff function introduced in Definition \ref{def:PPA_math}. In terms of notation, the PPA offtaker (the buyer) has the role of the investor $I$, while the wind power producer (the seller), has the role of the corporation $C$.

\begin{Proposition}\label{prop:cva_general}
The CVA at time $t$ of a wind power PPA with constant price $K$ and settlement and delivery times $T_j$ is 
\begin{align}
    \nonumber
    \mathrm{CVA(t,T)} &\approx  \mathrm{LGD_C}\sum_{i=0}^{n_b -1}\mathbb{E}^{\mathbb{Q}}\bigg[D(t,T_i)(\mathbb{E}^\mathbb{Q}[\Pi(T_i, T)|\mathcal{F}_{T_i}]^+) 
        e^{-\int_t^{T_{i+1}} \lambda_I(s)ds}\bigg(e^{-\int_t^{T_{i}} \lambda_C(s)ds} + \\
        \nonumber
        & \quad - e^{-\int_t^{T_{i+1}} \lambda_C(s)ds} \bigg)
        \bigg|\mathcal{F}_t \bigg] 
    \\
    \nonumber
    & \approx \mathrm{LGD_C}\sum_{i=0}^{n_b -1}\mathbb{E}^{\mathbb{Q}}\bigg[D(t,T_i)\bigg(\mathbb{E}^\mathbb{Q}\bigg[ \sum_{j=i}^n Q(T_j, W(T_j))\big(S(T_j) - K\big) \times
    \\
    \nonumber
     & \quad \times D(T_i,T_j) \bigg|\mathcal{F}_{T_i}\bigg]^+\bigg)  e^{-\int_t^{T_{i+1}} \lambda_I(s)ds}\bigg(e^{-\int_t^{T_{i}} \lambda_C(s)ds} - e^{-\int_t^{T_{i+1}} \lambda_C(s)ds} \bigg)
        \bigg|\mathcal{F}_t \bigg] 
    \\
    \nonumber
    & \approx \mathrm{LGD_C}\sum_{i=0}^{n_b -1}\mathbb{E}^{\mathbb{Q}}\bigg[D(t,T_i)\bigg(\mathbb{E}^\mathbb{Q}\bigg[\alpha \sum_{j=i}^n W(T_j)^3\mathbf{1}_{W(T_j)\in[a,b]} \big(S(T_j) - K\big) \times \\
    \label{eq:final_formula_cva_ppa}
    & \quad \times D(T_i,T_j) \bigg|\mathcal{F}_{T_i}\bigg]^+\bigg)  e^{-\int_t^{T_{i+1}} \lambda_I(s)ds} \bigg(e^{-\int_t^{T_{i}} \lambda_C(s)ds} - e^{-\int_t^{T_{i+1}} \lambda_C(s)ds} \bigg)
        \bigg|\mathcal{F}_t \bigg],
\end{align}
where $\bigcup_{i=0}^{n_b -1} [T_i, T_{i+1}) = [t, T_n)$.
\end{Proposition}
\begin{proof}
    The proof is limited to the substitution of Eq. \eqref{eq:pricing_wind_ppa_definition} inside Eq. \eqref{eq:CVA_pricing_formula_under_market_filtration_v2}.
\end{proof}
\begin{corollary}
The CVA at time $t$ under the Gaussian model can then be obtained by replacing 
$$ \mathbb{E}^\mathbb{Q}\bigg[\alpha\sum_{j=i}^n W(T_j)^3\mathbf{1}_{W(T_j)\in[a,b]} \big(S(T_j) - K\big)D(T_i,T_j) \bigg|\mathcal{F}_{T_i}\bigg] $$ 
with the explicit solution of 
$\mathbb{E}_t^\mathbb{Q}[\Pi(t,T,W(t),S(t))]$ from Proposition \ref{prop:val_gauss}, where $t=T_i$.
\end{corollary}

\begin{corollary}
The CVA at time $t$ under the jump-diffusion Ornstein-Uhlenbeck model can then be obtained by replacing 
$$\mathbb{E}^\mathbb{Q}\bigg[\alpha\sum_{j=i}^n W(T_j)^3\mathbf{1}_{W(T_j)\in[a,b]} \big(S(T_j) - K\big)D(T_i,T_j) \bigg|\mathcal{F}_{T_i}\bigg] $$ 
with the explicit solution of 
$\mathbb{E}_t^\mathbb{Q}[\Pi(t,T,W(t),S(t))]$ from Proposition \ref{prop:val_jump}, where $t=T_i$.
\end{corollary}

\begin{Proposition}\label{prop:dva_general}
The DVA at time $t$ of a wind power PPA with constant price $K$ and settlement and delivery times $T_j$ is 
\begin{align}
    \nonumber
    \mathrm{DVA(t,T)} &\approx  \mathrm{LGD_I}\sum_{i=0}^{n_b -1}\mathbb{E}^{\mathbb{Q}}\bigg[D(t,T_i)(\mathbb{E}^\mathbb{Q}[-\Pi(T_i, T)|\mathcal{F}_{T_i}]^+) 
        e^{-\int_t^{T_{i+1}} \lambda_C(s)ds}\bigg(e^{-\int_t^{T_{i}} \lambda_I(s)ds} + \\
    \nonumber
    & \quad - e^{-\int_t^{T_{i+1}} \lambda_I(s)ds} \bigg)
        \bigg|\mathcal{F}_t \bigg] 
    \\
    \nonumber
    & \approx \mathrm{LGD_I}\sum_{i=0}^{n_b -1}\mathbb{E}^{\mathbb{Q}}\bigg[D(t,T_i)\bigg(\mathbb{E}^\mathbb{Q}\bigg[- \sum_{j=i}^n Q(T_j, W(T_j))\big(S(T_j) - K\big) \times
    \\
    \nonumber
     & \quad \times D(T_i,T_j) \bigg|\mathcal{F}_{T_i}\bigg]^+\bigg)  e^{-\int_t^{T_{i+1}} \lambda_C(s)ds}\bigg(e^{-\int_t^{T_{i}} \lambda_I(s)ds} - e^{-\int_t^{T_{i+1}} \lambda_I(s)ds} \bigg)
        \bigg|\mathcal{F}_t \bigg] 
    \\
    \nonumber
    & \approx \mathrm{LGD_I}\sum_{i=0}^{n_b -1}\mathbb{E}^{\mathbb{Q}}\bigg[D(t,T_i)\bigg(-\mathbb{E}^\mathbb{Q}\bigg[\alpha \sum_{j=i}^n W(T_j)^3\mathbf{1}_{W(T_j)\in[a,b]} \big(S(T_j) - K\big)
    \\
    \label{eq:final_formula_dva_ppa}
     & \quad \times D(T_i,T_j) \bigg|\mathcal{F}_{T_i}\bigg]\bigg)^+ e^{-\int_t^{T_{i+1}} \lambda_C(s)ds}\bigg(e^{-\int_t^{T_{i}} \lambda_I(s)ds} - e^{-\int_t^{T_{i+1}} \lambda_I(s)ds} \bigg)
    \bigg|\mathcal{F}_t \bigg]
\end{align}
\end{Proposition}
\begin{proof}
    The proof is limited to the substitution of Eq. \eqref{eq:pricing_wind_ppa_definition} inside Eq. \eqref{eq:DVA_pricing_formula_under_market_filtration_v2}.
\end{proof}
\begin{corollary}
The DVA at time $t$ under the Gaussian model can then be obtained by replacing 
$$ \mathbb{E}^\mathbb{Q}\bigg[\alpha\sum_{j=i}^n W(T_j)^3\mathbf{1}_{W(T_j)\in[a,b]} \big(S(T_j) - K\big)D(T_i,T_j) \bigg|\mathcal{F}_{T_i}\bigg] $$
with the explicit solution of 
$\mathbb{E}_t^\mathbb{Q}[\Pi(t,T,W(t),S(t))]$ from Proposition \ref{prop:val_gauss}, where $t=T_i$.
\end{corollary}

\begin{corollary}
The DVA at time $t$ under the jump-diffusion Ornstein-Uhlenbeck model can then be obtained by replacing 
$$ \mathbb{E}^\mathbb{Q}\bigg[\alpha\sum_{j=i}^n W(T_j)^3\mathbf{1}_{W(T_j)\in[a,b]} \big(S(T_j) - K\big)D(T_i,T_j) \bigg|\mathcal{F}_{T_i}\bigg] $$
with the explicit solution of 
$\mathbb{E}_t^\mathbb{Q}[\Pi(t,T,W(t),S(t))]$ from Proposition \ref{prop:val_jump}, where $t=T_i$.
\end{corollary}

\subsection{Pricing with CVA and DVA}\label{sec:K_with_BVA}
In Section \ref{sec:risk_neutral} the risk-neutral approach to the pricing of the PPA is introduced. Its main limitation is however highlighted, and it resides in the fact that an OTC contract must account for counterparty credit risk. After having provided explicit formulas for the CVA and DVA of such contract in Section \ref{sec:the_counterparty_credit_risk_in_ppa}, we therefore propose an updated, counterparty-risk-adjusted pricing formula in line with Eq. \eqref{eq:ppa_pricing_formula}.

\begin{definition}\label{def:credit_risk:price}
\textbf{Counterparty-risk-adjusted price of a wind power PPA}\newline
The counterparty-risk-adjusted price of a wind power PPA contract is the constant payment $\bar{K}  \in \mathbb{R}$ that makes the discounted expectation of the counterparty-risk-adjusted PPA payoff equal to zero, i.e. 
\begin{align}
    \mathbb{E}^{\mathbb{Q}}[\bar\Pi(t,T,\bar{K})|\mathcal{F}_t] &= \mathbb{E}^{\mathbb{Q}}[\Pi(t,T,\bar{K})|\mathcal{F}_t] + \mathrm{DVA}(t,T,\bar{K})-\mathrm{CVA}(t,T,\bar{K}) =0, \label{eq:counterparty_adjusted_price}
\end{align}
where $\mathbb{E}^{\mathbb{Q}}[\bar\Pi(t,T,\bar{K})|\mathcal{F}_t]$ is the counterparty-risk-adjusted PPA payoff, written as an explicit function of $K$.
\end{definition}
It is however not possible to solve Eq. \eqref{eq:counterparty_adjusted_price} for $\bar{K}$ analytically, either with the Gaussian or the Jump model, due to the non-linearity in the CVA and DVA formulas. Therefore, only numerical optimization methods can be used.

\section{An empirical illustration of CVA and DVA}\label{sec:empirical}
In this section we offer an application of the methodology developed so far. We present an example transaction between an offtaker (in our case, ENEL Distribution S.p.A.) and a wind power producer (ENI Plenitude S.p.A.). We choose these two companies because of the availability and liquidity of their CDS data, from which we can extract the default probabilities required to calibrate hazard rate models. The aim of this section is not to analyze a specific OTC contract, but to demonstrate the role and impact of the value adjustment, by using realistic input parameters. 

\subsection{Data}
The data required in this section includes average daily wind speed, which has been downloaded from \url{https://power.larc.nasa.gov/data-access-viewer/} for the location with latitude 41.2087 and longitude 16.4016. The time series of electricity prices corresponds to the Italian unique national price (PUN) and can be downloaded directly from the GME (Gestore del Mercato Elettrico) at \url{https://gme.mercatoelettrico.org/it-it/Home/Esiti/Elettricita/MGP/Esiti/PUN}. Finally, CDS data has been downloaded from Refinitiv Datastream.

\subsection{Pricing}
We being by calibrating the parameters of the models in Eq. \eqref{eq:wind_sol_gauss}, \eqref{eq:lambda_W}, \eqref{eq:electricity_sol_gauss}, \eqref{eq:lambda_S}, \eqref{eq:implicit_SDE_wind_intensity}, \eqref{eq:implicit_SDE_price_electricity}, and \eqref{eq:electricity_wind_dependence_condition}. They are reported in Table \ref{tab:gaussian_case_parameters}, for the Gaussian model, and Table \ref{tab:levy_ou__case_parameters}, for the jump-diffusion Ornstein-Uhlenbeck model.

\begin{table}[h!]
\centering
\begin{tabular}{lcc}
\hline
 & Electricity price $S(t)$ &  Wind speed $W(t)$ \\
\hline
$\mu_i$    & 4.32981      & 1.27079       \\
$a_i$              & 0.0236305    & -0.0250245    \\
$b_i$              & 0.0501226    & 0.0902814     \\
$\kappa_i$          & 0.394135    & 0.746079    \\
$\theta_i$          & -0.00026124 & 0.00160128  \\
$\sigma_i$          & 0.039872    & 0.124079    \\
\hline
$\rho$            & \multicolumn{2}{c}{-0.054} \\
\hline
\end{tabular}
\caption{Model parameters for the two underlying variables in the Gaussian case.}
\label{tab:gaussian_case_parameters}
\end{table}

The parameters $\kappa_i$, $\theta_i$, and $\sigma_i$, $i \in \{W, S \}$, correspond to the stochastic components of wind speed and electricity prices, modeled as Ornstein-Uhlenbeck processes, where $\kappa_i$ is the speed of mean reversion, $\theta_i$ is the long-term mean level, and $\sigma_i$ is the volatility. The parameter $\rho$ denotes the correlation between the two stochastic components.
The parameters $\mu_i$, $a_i$, and $b_i$, $i \in \{W, S \}$, define the deterministic seasonal components of wind speed and electricity prices modeled via the sine-cosine expansion in \eqref{eq:lambda_W} and \eqref{eq:lambda_S}.

\begin{table}[h!]
\centering
\begin{tabular}{lcc}
\hline
                 & Electricity price $S(t)$          & Wind speed $W(t)$  \\
\hline
$\mu_i$            & 4.32981     & 1.27079     \\
$a_i$              & 0.0236305   & -0.0250245  \\
$b_i$              & 0.0501226   & 0.0902814   \\
$\kappa_i$         & 0.0093718   & 0.781841    \\
$\theta_i$         & -0.281307   & 4.34914e-05 \\
$\sigma_i$         & 0.00688945  & 0.0471666   \\
\hline
$\lambda_{\text{jump},i}$ & 0.0130844   & 0.0119723   \\
$\mu_{\text{jump},i}$     & 0           & 0           \\
$\sigma_{\text{jump},i}$  & 0.0145784   & 0.000574602 \\
\hline
$\rho$           & \multicolumn{2}{c}{-0.05} \\
\hline
\end{tabular}
\caption{Model parameters for the two underlying variables in the jump-diffusion Ornstein-Uhlenbeck
model.}
\label{tab:levy_ou__case_parameters}
\end{table}

Table \ref{tab:levy_ou__case_parameters} has the additional parameters $\lambda_{jump,i}$, $\mu_{jump,i}$, and $\sigma_{jump,i}$, $i \in \{ W,S\}$, are the jump part parameters. They represent the jump arrival intensity, expected jump size, and jump volatility, respectively.
\\
The starting points for the deseasonalized processes $X(t)$ and $Y(t)$ used in the simulations containes in this section are reported in Table \ref{tab:initial_values}. The efficiency parameter of the wind-power production function in Eq. \eqref{eq:wind_power_prod} is set to $\alpha=1$, in order to provide results in a unitary scale.

\begin{table}[h!]
\centering
\begin{tabular}{lc}
\hline
Parameter & Value \\
\hline
$e^{X_0}$ & $ 1.0$ \\
$e^{Y_0}$  & $ 1.0$ \\
$\gamma_C(0)$ & 0.0091 \\
$\gamma_I(0)$ & 0.0157 \\
\hline
\end{tabular}
\caption{Initial values for the deseasonalized processes $X_t$, $Y_t$ and the hazard rates}
\label{tab:initial_values}
\end{table}

The hazard rates are calibrated on CDS data for ENEL and ENI, the assumed offtaker and producer of this setting. The parameters of the processes in Eq. \eqref{eq:hazard_rate_buyer} and \eqref{eq:hazard_rate_seller} are reported in Table \ref{tab:cir_parameters}. They are $k_j$, $\theta_j$ and $\sigma_j$, $j \in \{C, I \}$, and they correspond to the speed of mean reversion, long-term mean level, and volatility of the CIR processes, respectively.

\begin{table}[h!]
\centering
\begin{tabular}{lcc}
\hline
 & Producer ($j=C$) & Offtaker ($j=I$) \\
\hline
$k_j$    & 0.021  & 0.064 \\
$\theta_j$   & 1.702  & 1.43 \\
$\sigma_j$    & 0.201  & 0.4 \\
\hline
\end{tabular}
\caption{Cox-Ingersoll-Ross (CIR) model parameters for counterparty and investor}
\label{tab:cir_parameters}
\end{table}

We finally use the estimated parameters to perform PPA valuation (following Eq. \eqref{eq:final_ppa_gaussian} and \eqref{eq:final_ppa_jump}), pricing (following Eq. \eqref{eq:K_gauss} and \eqref{eq:K_jump}), and to compute the CVA, DVA, and BVA (following Eq. \eqref{eq:final_formula_cva_ppa} and \eqref{eq:final_formula_dva_ppa}) of a number of PPA contracts between the selected offtaker and producer. We do so for multiple different maturities and for both models: Gaussian and jump diffusion Ornstein-Uhlenbeck.   

\begin{figure}[h!]
    \centering
    \includegraphics[width=0.75\linewidth]{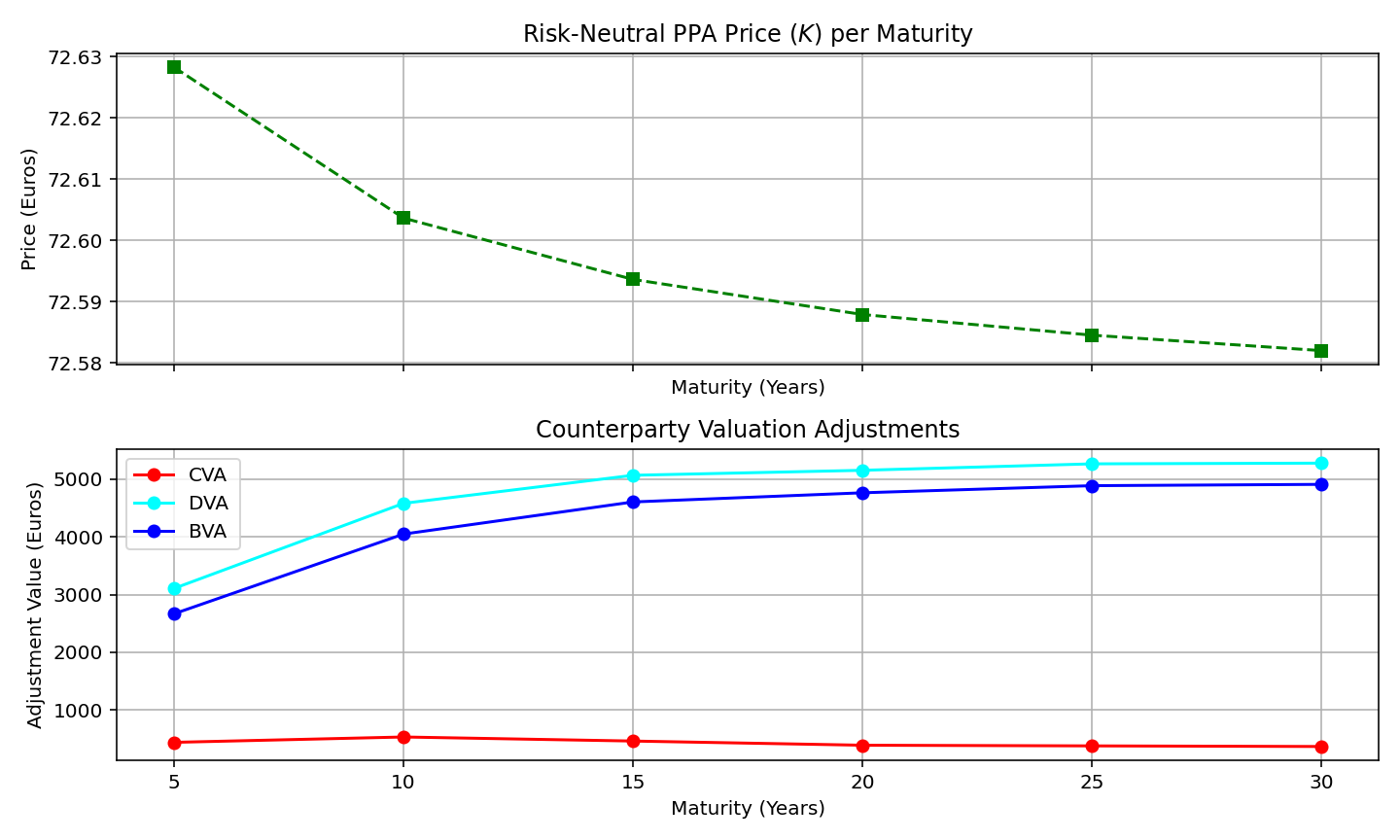}
    \caption{PPA fair price and XVA for the Gaussian case}
    \label{fig:fair_price_ppa_gaussian}
\end{figure}

\newpage

\begin{figure}[h!]
    \centering
    \includegraphics[width=0.75\linewidth]{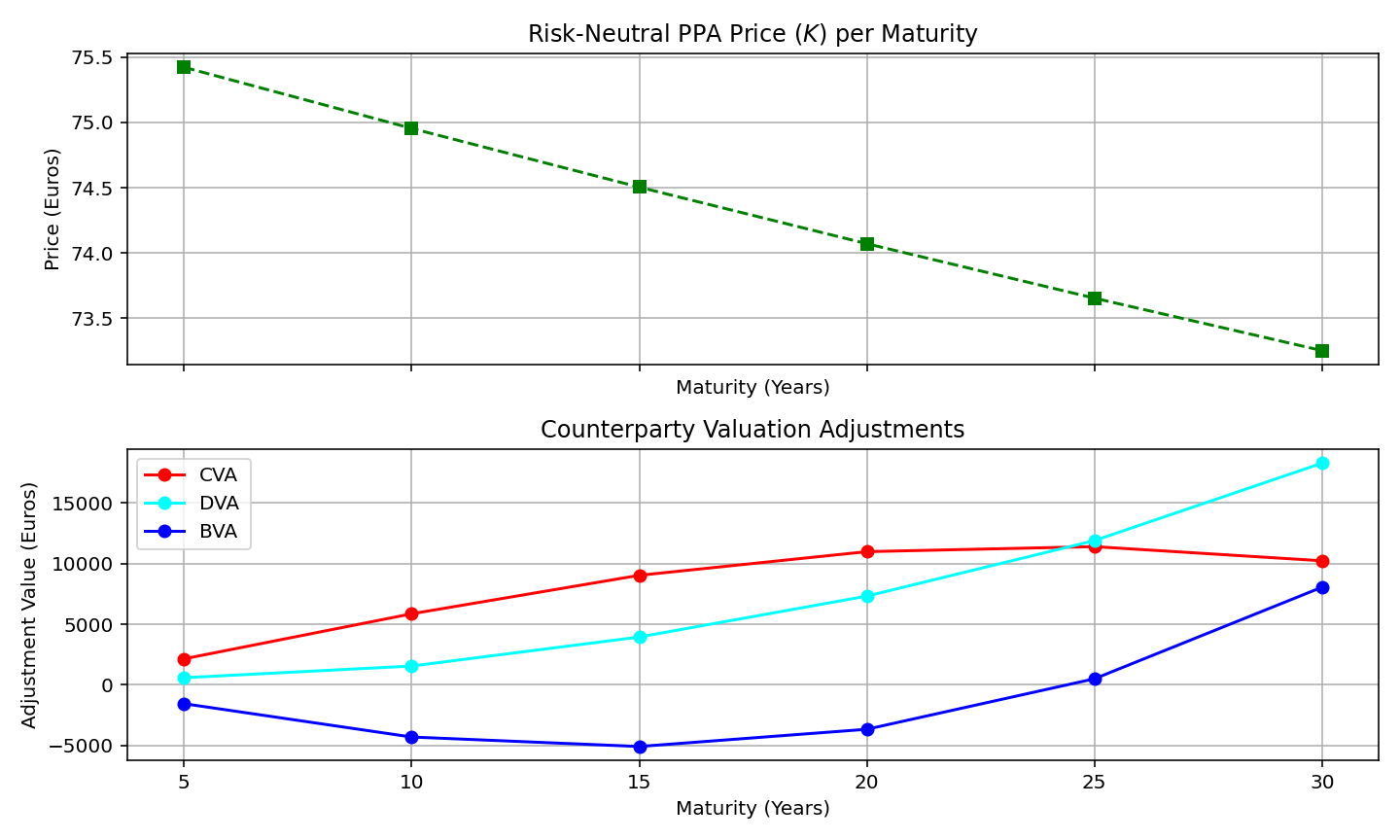}
    \caption{PPA fair price and XVA for the Levy case}
    \label{fig:fair_price_ppa_levy}
\end{figure}

The results are in Figures \ref{fig:fair_price_ppa_gaussian} and \ref{fig:fair_price_ppa_levy}. For both models we get similar PPA prices, which interestingly are also comparable with the results (for the same country, Italy) found by \cite{MENDICINO2019113577} using the LCOE-based valuation method. However, when  evaluating the BVA, we see that it is always far from zero. This highlights the main motivation behind this work: when dealing with OTC contracts and not with standardized and market-cleared financial products, defining the fairness of the contract only on the basis of the risk-neutral price will not reflect all the risky components of the transaction, as counterparty risk is excluded. Therefore, the risk-neutral price is not a guarantee of a zero-valued position, as the two parties must consider that the actual value of their positions must account for CVA and DVA. 
\\
We also finally notice that the BVAs differ noticeably depending on the underlying model assumptions. This indicates the great impact of appropriate model selection when evaluating PPAs, which must be tailored to the time series at hand. 

\section{Policy implications for the EIB PPA guarantee scheme}\label{sec:policy}
In June 2025, the European Investment Bank (EIB) approved a 500 million Euro guarantee programme to support corporate power purchase agreements for renewable energy.\footnote{European Investment Bank, ``PAN-EU POWER PURCHASE AGREEMENT GUARANTEE", Project No. 20250202, approved 19 June 2025. \url{https://www.eib.org/en/projects/all/20250202}}. The scheme was created to provide counter-guarantees that back those issued by commercial financial intermediaries, with the goal of decreasing the risk of long-term electricity contracts for offtakers. Among them, particular targets are small and medium-sized enterprises and energy-intensive industries that lack strong credit ratings.\footnote{Pexapark, ``EIB's PPA Guarantee System under the Clean Industrial Deal: How it works'', March 2025. \url{https://pexapark.com/blog/eibs-ppa-guarantee-system-under-the-clean-industrial-deal/}}. By absorbing part of the counterparty risk, the EIB helps offtakers access green energy at a stable price and enables renewable energy producers to secure a predictable demand for their output. 
\\
The work presented in this paper, evaluating and providing formulas to calculate the CVA and DVA of wind-power PPAs, can play a supporting role for policy making. These formulas incorporate the unique cash-flow characteristics of wind power production: volume uncertainty, price volatility and the correlation between wind speed and electricity prices. Within the guarantee programme, the EIB could use a netting approach (the difference between CVA and DVA) as a measure of the net credit risk that a guarantee would actually need to cover, in order to avoid over-guaranteeing bilateral risk. A guarantee that covers only the offtaker's CVA would ignore the fact that a producer with a high DVA (i.e., a high own-default risk) naturally reduces the expected loss. Consequently, the EIB could calibrate its guarantee fee based on the net amount and design tiered guarantee structures (such as first-loss versus full coverage) that reflect the offsetting effects of the two sides of the PPA.
\\
A second policy implication concerns the dynamic adjustment of guarantee terms. CVA and DVA evolve overtime, together with changes in credit ratings, market conditions and the financial health of both counterparties. This work could provide a quantitative tool for the EIB to monitor these adjustments in real time. An increase in the offtaker's CVA above a predetermined threshold, while the producer's DVA remains low, could lead to an automatic increase on the part of the EIB of the guarantee fee, or to a request for additional collateral from the offtaker. Conversely, if the producer's DVA rises sharply (indicating that the producer is becoming distressed) the EIB can reduce the guarantee coverage, because the producer is less likely to perform and call on the guarantee. A dynamic risk management approach could make use of CVA and DVA as active signals that keep the guarantee scheme aligned with the evolving risk profile of the PPA over its entire lifetime. The impact of this approach could be highlighted by comparing the total guarantee cost and expected loss under the dynamic policy versus a static one.
\\
Details about the mathematical and financial modelling behind the guarantee scheme are not currently disclosed, which prevents a reliable and sound illustration of the above-mentioned applications. We reserve a deeper analysis of the impact of our proposal for when such information is available.

\section{Conclusions}\label{sec:conclusions}

The objective of this work was to address a gap in the current literature: the quantification of counterparty credit risk in commercial renewable Power Purchase Agreements. While PPAs are widely recognised for their role in stabilising revenues and facilitating the financing of renewable energy plants, their bilateral over-the-counter structure exposes both parties to credit risk. This dimension has currently remained underexplored, particularly given the double uncertainty inherent in renewable generation, caused by unpredictable electricity prices and wind speed. 
\\
In this work, we first define the risk neutral price of a PPA under the assumption of stochastic electricity prices and stochastic wind power production. This framework captures the correlated joint dynamics of the underlying risk factors and provides a foundation for subsequent credit risk analysis. We then link the risk neutral price to the standard counterparty credit risk framework by specifying both the Credit Valuation Adjustment and the Debit Valuation Adjustment. We adapt the stochastic wind power production function to account explicitly for its link with wind speed. The resulting pricing and valuation problem leads to analytic, in the Gaussian setting, and semi-analytic ones, in the case of a jump-diffusion model for the underlying quantities.
\\
CVA and DVA provide market participants with a transparent metric for pricing counterparty credit risk into long term PPAs, thereby informing contract negotiation and collateral requirements. These adjustments are essential tools for proper risk management and allow firms to monitor exposures, set credit limits and allocate capital appropriately. Recent policy initiatives, such as the European Investment Bank's pilot guarantee scheme, aim to mitigate credit risk for PPA market players. With this work, we offer a practical tool for the valuation of counterparty credit risk adjustments, which are not only needed for the internal credit risk assessment of producers and offtakers that are not covered by the scheme, but that can also serve as a benchmark for credit risk mitigation tools, which include the very own pilot scheme by the EIB. 
We hope to contribute to the broader efforts of supporting the sustainable expansion of renewable energy through sound financial risk management.  
\\
As for the limitations of this work, the first concerns the choice of the stochastic models for the underlying random variables, which we show to have a great impact on prices and valuation adjustments. Future research could focus on the long-term horizon of these contracts, which might call for models with parameters that evolve overtime for both wind and electricity. An additional avenue to be pursued would involve the adaptation of the PPA payoff function. It could include more elaborate structures for the price to be paid by the offtaker, which can be a non-constant parameter, or a mismatch between delivery and settlement periods. A final path for future research could extend the framework to the modelling of photovoltaic PPAs, which would require entirely different models for the underlying weather driver (solar irradiance) and for the corresponding power production function.

\newpage

\bibliographystyle{plainnat}
\bibliography{references}

\newpage

\begin{appendix}

\section{PPA Pricing with the Gaussian Model}
\subsection{PPA Valuation}\label{app:ppa_gaussian_val}
\begin{proof}
Proof of Proposition \ref{prop:val_gauss}.\\
Let us assume deterministic interest rates and focus on the single UPPA price for maturity $T$, which from Def. \ref{def:UPPA} and Eq. \eqref{eq:pricing_wind_ppa_definition} is
\begin{align}\label{eq:uppa_app1}
    UPPA(t, T) &= D(t,T) \big\{\mathbb{E}_t^\mathbb{Q}[\alpha W(T)^3 \mathbf{1}_{\{W(T) \in [a,b]\}}S(T)] \notag \\ 
    &\quad - K \mathbb{E}_t^\mathbb{Q}[\alpha W(T)^3 \mathbf{1}_{\{W(T) \in [a,b]\}}]\big\}.
\end{align}

We begin by focusing on the second expected value, i.e.
\begin{align*}
    \mathbb{E}_t^\mathbb{Q}[\alpha W(T)^3\mathbf{1}_{\{W(T) \in [a,b]\}}]  &= \alpha e^{3\Lambda_W(T) + 3Y(t)e^{-\kappa_W(T-t)} + 3\theta_W(1-e^{-\kappa_W(T-t)})} \\
    &\quad  \times \mathbb{E}_t^\mathbb{Q}\bigg[ \exp\bigg\{3\sigma_W\int_t^Te^{-\kappa_W(T-s)} dZ^W(s)  \bigg\} \mathbf{1}_{\{W(T) \in [a,b]\}} \bigg]. 
\end{align*}
We divide and multiply by $\exp\{\frac{9}{2}\sigma_W^2 \int_t^T e^{-2\kappa_W(T-s)ds }\}= \exp\{\frac{9\sigma_W^2}{4\kappa_W} (1-e^{-2\kappa_W(T-t)})\}$ inside the expectation, in order to highlight 
$$ \frac{d\mathbb{\tilde{Q}}}{d\mathbb{Q}} \Bigg|_{T}= \frac{e^{3\sigma_W\int_t^Te^{-\kappa_W(T-s)} dZ^W(s)}}{e^{\frac{9}{2}\sigma_W^2 \int_t^T e^{-2\kappa_W(T-s)ds }}} $$ 
as a Radon-Nikodym derivative. We then get 
\begin{align*}
    \mathbb{E}_t^\mathbb{Q}[\alpha W(T)^3\mathbf{1}_{\{W(T) \in [a,b]\}}]  &= \alpha e^{3\Lambda(T) + 3Y(t)e^{-\kappa_W(T-t)} + 3\theta_W(1-e^{-\kappa_W(T-t)}) } \notag \\
    &\quad \times e^{\frac{9\sigma_W^2}{4\kappa_W} (1-e^{-2\kappa_W(T-t)})}\mathbb{E}_t^\mathbb{\tilde{Q}}\bigg[\mathbf{1}_{\{W(T) \in [a,b]\}} \bigg]  \notag \\
     &= \alpha e^{3\Lambda(T) + 3Y(t)e^{-\kappa_W(T-t)} + 3\theta_W(1-e^{-\kappa_W(T-t)})}  \notag \\
    &\quad  \times e^{\frac{9\sigma_W^2}{4\kappa_W} (1-e^{-2\kappa_W(T-t)})} \tilde{\mathbb{Q}}(W(T)\in[a,b]). 
\end{align*}
Then, since $3\sigma_W e^{-\kappa_W (T-t)} dZ^W(s)$ is Gaussian, we can recover the new density under the new measure induced by the Esscher Transform $\frac{e^{\mu Xf_X(x)}}{\mathbb{E}[e^{\mu X}]}$, where $\mu=1$ and $X = \int_t^T3\sigma_W e^{-\kappa_W (T-s)} dZ^W(s)$.
We then recover
\begin{align*} 
 \tilde{\mathbb{Q}}(W(T)\in[a,b]) & = \Phi( \tilde{b}) - \Phi( \tilde{a}),
\end{align*}
where $\Phi(\cdot)$ is the Gaussian CDF under the the new measure and where
$$\tilde{a} = \frac{\ln(a) -  \Lambda_W(T) - Y(t)e^{-\kappa_W(T-t)} - \theta_W(1-e^{-\kappa_W(T-t)}) -\frac{3\sigma_W^2}{2\kappa_W}(1-e^{-2\kappa_W(T-t)})}{\sigma_W\sqrt{\frac{1-e^{-2\kappa_W(T-t)}}{2\kappa_W}}} $$
and
$$\tilde{b} =\frac{\ln(b) -  \Lambda_W(T) - Y(t)e^{-\kappa_W(T-t)} - \theta_W(1-e^{-\kappa_W(T-t)}) -\frac{3\sigma_W^2}{2\kappa_W}(1-e^{-2\kappa_W(T-t)})}{\sigma_W\sqrt{\frac{1-e^{-2\kappa_W(T-t)}}{2\kappa_W}}}. $$
We finally obtain
\begin{align}
    \mathbb{E}_t^\mathbb{Q}[\alpha W(T)^3\mathbf{1}_{\{W(T) \in [a,b]\}}]  &= \alpha e^{3\Lambda_W(T) + 3Y(t)e^{-\kappa_W(T-t)} + 3\theta_W(1-e^{-\kappa_W(T-t)})}  \notag \\
    &\quad  \times e^{\frac{9\sigma_W^2}{4\kappa_W} (1-e^{-2\kappa_W(T-t)})} ( \Phi(\tilde{b}) - \Phi(\tilde{a})). \label{eq:second_expectation}
\end{align}

We now move to the first expected value in Eq. \eqref{eq:uppa_app1}, and, by substituting in Eq. \eqref{eq:wind_sol_gauss} and \eqref{eq:electricity_sol_gauss}, we recover
\begin{align*}
    \mathbb{E}_t^\mathbb{Q}[\alpha W(T)^3 \mathbf{1}_{\{W(T) \in [a,b]\}}S(T)] &=  \alpha e^{3\Lambda_W(T) + 3Y(t)e^{-\kappa_W(T-t)} + 3\theta_W(1-e^{-\kappa_W(T-t)}) + \Lambda_S(T)} \\
    &\quad \times e^{X(t)e^{-\kappa_S(T-t)} + \theta_S(1-e^{-\kappa_S(T-t)})} \\
    & \quad \times \mathbb{E}_t^\mathbb{Q}\bigg[ \exp\bigg\{3\sigma_W\int_t^Te^{-\kappa_W(T-s)} dZ^W(s) 
    \\
    & \quad
    + \sigma_S\int_t^T e^{-\kappa_S(T-u)}dZ^S(u) \bigg\} \mathbf{1}_{\{W(T) \in [a,b]\}} \bigg]
\end{align*}

We remark that the two Brownian Motions, $Z^W$ and $Z^S$, admit the following representation
\begin{equation*}
\begin{pmatrix}
        dZ^W(s) \\
        dZ^S(s)
    \end{pmatrix} 
    = \begin{pmatrix}
        1 & 0 \\
        \rho & \sqrt{1 - \rho^2}
    \end{pmatrix}
    \begin{pmatrix}
        dB^W(s) \\
        dB^S(s)
    \end{pmatrix}
\end{equation*} 
Where $B^W(t)$ and $B^S(t)$ are independent standard Brownian Motions. We can then rewrite the conditional expected value as
\begin{align*}
    \mathbb{E}_t^\mathbb{Q}\bigg[ &\exp\bigg\{3\sigma_W\int_t^Te^{-\kappa_W(T-s)} dZ^W(s) + \sigma_S\int_t^T e^{-\kappa_S(T-u)}dZ^S(u) \bigg\} \mathbf{1}_{\{W(T) \in [a,b]\}} \bigg] = \\
    & = \mathbb{E}_t^\mathbb{Q} \bigg[ \exp\bigg\{3\sigma_W\int_t^Te^{-\kappa_W(T-s)} dB^W(s) + \sigma_S\int_t^T e^{-\kappa_S(T-u)}d(\rho B^W(u)  \\
    & \quad + \sqrt{1-\rho^2}B^S(u)) \bigg\} \mathbf{1}_{\{W(T) \in [a,b]\}} \bigg] \\
    & = \mathbb{E}_t^\mathbb{Q}\bigg[\exp\bigg\{\sigma_S\int_t^T e^{-\kappa_S(T-u)}\sqrt{1-\rho^2} dB^S(u)\bigg\} \bigg] \\
    &\quad \times  \mathbb{E}_t^\mathbb{Q}\bigg[\exp\bigg\{\int_t^T \bigg(3\sigma_We^{-\kappa_W(T-s)} + \rho\sigma_S e^{-\kappa_S(T-s)} \bigg)dB^W(s) \bigg\} \mathbf{1}_{\{W(T) \in [a,b]\}} \bigg]
\end{align*}
Since $B^S$ and $B^W$ are independent.
We then remark that
$$\mathbb{E}_t^\mathbb{Q}\bigg[\exp\bigg\{\sigma_S\int_t^T e^{-\kappa_S(T-u)}\sqrt{1-\rho^2} dB^S(u)\bigg\} \bigg] = e^{ \frac{\sigma_S^2(1-\rho^2)}{4\kappa_S}(1-e^{-2\kappa_S(T-t)})}$$ and substitute it in the above equation, yielding
\begin{align*}
    \mathbb{E}_t^\mathbb{Q}&\bigg[\exp\bigg\{3\sigma_W\int_t^Te^{-\kappa_W(T-s)} dZ^W(s) + \sigma_S\int_t^T e^{-\kappa_S(T-u)}dZ^S(u) \bigg\} \mathbf{1}_{\{W(T) \in [a,b]\}} \bigg] = \\
    & =  e^{ \frac{\sigma_S^2(1-\rho^2)}{{4}\kappa_S}(1-e^{-{2}\kappa_S(T-t)})}   \mathbb{E}_t^\mathbb{Q}\bigg[\exp\bigg\{\int_t^T \bigg(3\sigma_We^{-\kappa_W(T-s)} + \rho\sigma_S e^{-\kappa_S(T-s)} \bigg)dB^W(s) \bigg\} \\
    &\quad \times \mathbf{1}_{\{W(T) \in [a,b]\}} \bigg].
\end{align*}
We next divide and multiply by $\mathbb{E}_t^\mathbb{Q}\bigg[\exp\bigg\{\int_t^T \bigg(3\sigma_We^{-\kappa_W(T-u)} + \rho\sigma_S e^{-\kappa_S(T-u)} \bigg)dB^W(u) \bigg\} \bigg]$, in order to isolate the following Esscher transform for a measure change to a new equivalent measure $\mathbb{\bar{Q}}$, i.e.
$$\frac{d\mathbb{\bar{Q}}}{d\mathbb{Q}} \Bigg|_T=\frac{\exp\bigg\{\int_t^T \bigg(3\sigma_We^{-\kappa_W(T-u)} + \rho\sigma_S e^{-\kappa_S(T-u)} \bigg)dB^W(u) \bigg\}}{\mathbb{E}_t^\mathbb{Q}\bigg[\exp\bigg\{\int_t^T \bigg(3\sigma_We^{-\kappa_W(T-u)} + \rho\sigma_S e^{-\kappa_S(T-u)} \bigg)dB^W(u) \bigg\} \bigg]}.$$
This then yields
\begin{align*}
    \mathbb{E}_t^\mathbb{Q}\bigg[ &\exp\bigg\{3\sigma_W\int_t^Te^{-\kappa_W(T-s)} dZ^W(s) + \sigma_S\int_t^T e^{-\kappa_S(T-u)}dZ^S(u) \bigg\}  \mathbf{1}_{\{W(T) \in [a,b]\}} \bigg] = \\
    &\quad =e^{\frac{\sigma_S^2(1-\rho^2)}{4\kappa_S}(1-e^{-2\kappa_S(T-t)})} \mathbb{E}_t^\mathbb{Q}\bigg[\exp\bigg\{\int_{t}^{T} \bigg(3\sigma_We^{-\kappa_W(T-u)} + \rho\sigma_S e^{-\kappa_S(T-u)} \bigg)dB^W(u) \bigg\} \\
    & \quad \quad \times \mathbf{1}_{\{W(T) \in [a,b]\}} \bigg]   \\
    &\quad = e^{\frac{\sigma_S^2(1-\rho^2)}{4\kappa_S}(1-e^{-2\kappa_S(T-t)})} \mathbb{E}_t^\mathbb{Q}\bigg[\exp\bigg\{\int_t^T \bigg(3\sigma_We^{-\kappa_W(T-u)} + \rho\sigma_S e^{-\kappa_S(T-u)} \bigg)dB^W(u) \bigg\} \bigg] \\
    &\quad \quad \times    \mathbb{E}_t^\mathbb{Q}\bigg[\frac{\exp\bigg\{\int_t^T \bigg(3\sigma_We^{-\kappa_W(T-u)} + \rho\sigma_S e^{-\kappa_S(T-u)} \bigg)dB^W(u) \bigg\}}{\mathbb{E}_t^\mathbb{Q}\bigg[\exp\bigg\{\int_t^T \bigg(3\sigma_We^{-\kappa_W(T-u)} + \rho\sigma_S e^{-\kappa_S(T-u)} \bigg)dB^W(u) \bigg\} \bigg]} \mathbf{1}_{\{W(T) \in [a,b]\}} \bigg] 
    \\
    &\quad = e^{\frac{\sigma_S^2(1-\rho^2)}{{4}\kappa_S}(1-e^{-{2}\kappa_S(T-t)}) +\frac{9 \sigma_W^2}{4 \kappa_W}(1-e^{-2\kappa_W (T-t)}) +\frac{\rho^2 \sigma_S^2}{4\kappa_S}(1-e^{-2\kappa_S(T-t)})} \\
    & \quad \quad \times e^{+\frac{3\rho\sigma_W\sigma_S}{\kappa_W+\kappa_S}(1-e^{-(\kappa_W+\kappa_S)(T-t)})} \bar{\mathbb{Q}}_t (W(T) \in [a,b]), 
\end{align*}
by exploiting the fact that 
\begin{align*}
&\mathbb{E}_t^\mathbb{Q}\bigg[\exp\bigg\{\int_t^T \bigg(3\sigma_We^{-\kappa_W(T-u)} + \rho\sigma_S e^{-\kappa_S(T-u)} \bigg)dB^W(u) \bigg\} \bigg] = e^{ \frac{9 \sigma_W^2}{4 \kappa_W}(1-e^{-2\kappa_W (T-t)}) }\\
&\quad \times e^{\frac{\rho^2 \sigma_S^2}{4\kappa_S}(1-e^{-2\kappa_S(T-t)}) + \frac{3\rho\sigma_W\sigma_S}{\kappa_W+\kappa_S}(1-e^{-(\kappa_W+\kappa_S)(T-t)})}.
\end{align*}
With this measure change, by Girsanov's theorem, we define $\bar{B}^W_s$ as a $\mathbb{\bar{Q}}$-Brownian Motion satisfying
$$ d\bar{B}^W(t) = dB^W(t) - (3\sigma_W e^{-\kappa_W(T-t)} +\rho \sigma_S e^{-\kappa_S(T-t)}) dt.$$
We therefore have that, under $\mathbb{\bar{Q}}$, 
\begin{align*}
W(T) &= \exp\left\{\Lambda_W(T)  + Y(t)e^{-\kappa_W (T-t)} + \theta_W(1-e^{-\kappa_W (T-t)}) + \right.  \\
&\quad +  \sigma_W \int_t^T e^{-\kappa_W(T-s)} d\bar{B}^W(s)   3\sigma_W^2 \int_t^T e^{-2\kappa_W(T-s)}ds + \\
& \quad \left. + \rho \sigma_S \sigma_W \int_t^T e^{-(\kappa_W+\kappa_S)(T-s)}ds) \right\},
\end{align*}
which leads to
\begin{align*}
\bar{\mathbb{Q}}_t (W(T) \in [a,b]) &= \Phi(\bar{b})- \Phi(\bar{a}),     
\end{align*}
where 
\begin{align*}
&\bar{a} = \bigg(\ln(a) - \Lambda_W(T_j)  - Y(t)e^{-\kappa_W (T_j-t)} - \theta_W(1-e^{-\kappa_W (T_j-t)}) - \frac{3\sigma^2_W}{2\kappa_W} (1-e^{-2\kappa_W(T_j-t)})\\ 
&\quad -\frac{\rho \sigma_S \sigma_W}{\kappa_W + \kappa_S} (1-e^{-(\kappa_W+\kappa_S)(T_j-t)}) \bigg)
\bigg/ {\sigma_W  \sqrt{\frac{1}{2\kappa_W}(1- e^{-2\kappa_W(T_j-t)})}} 
\end{align*}
and
\begin{align*}
&\bar{b} = \bigg(\ln(b) - \Lambda_W(T_j)  - Y(t)e^{-\kappa_W (T_j-t)} - \theta_W(1-e^{-\kappa_W (T_j-t)}) - \frac{3\sigma^2_W}{2\kappa_W} (1-e^{-2\kappa_W(T_j-t))}) \\
&\quad -\frac{\rho \sigma_S \sigma_W}{\kappa_W + \kappa_S} (1-e^{-(\kappa_W+\kappa_S)(T_j-t)}) \bigg)
\bigg/ {\sigma_W  \sqrt{\frac{1}{2\kappa_W}(1- e^{-2\kappa_W(T_j-t)})}}.
\end{align*}

The first expected value then becomes
\begin{align}
    \mathbb{E}_t^\mathbb{Q}&[\alpha W(T)^3 \mathbf{1}_{\{W(T) \in [a,b]\}}S(T)] = \alpha e^{3\Lambda_W(T) + 3Y(t)e^{-\kappa_W(T-t)} + 3\theta_W(1-e^{-\kappa_W(T-t)}) + \Lambda_S(T)} \notag \\
    & \quad \times e^{X(t)e^{-\kappa_S(T-t)}  + \theta_S(1-e^{-\kappa_S(T-t)})+ \frac{\sigma_S^2(1-\rho^2)}{4\kappa_S}(1-e^{-2\kappa_S(T-t)}) + \frac{9 \sigma_W^2}{4 \kappa_W}(1-e^{-2\kappa_W (T-t)}) } \notag \\
    &\quad \times e^{\frac{\rho^2 \sigma_S^2}{4\kappa_S}(1-e^{-2\kappa_S(T-t)})+ \frac{3\rho\sigma_W\sigma_S}{\kappa_W+\kappa_S}(1-e^{-(\kappa_W+\kappa_S)(T-t)})} \big( \Phi(\bar{b})- \Phi(\bar{a})  \big)
    \label{eq:first_expectation} 
\end{align}
Finally, using Eq. \eqref{eq:second_expectation} and Eq. \eqref{eq:first_expectation} inside Eq. \eqref{eq:uppa_app1} leads to
\begin{align*}
    UPPA(t, T) &= D(t,T)\Big\{\mathbb{E}_t^\mathbb{Q}[\alpha W(T)^3 \mathbf{1}_{\{W(T) \in [a,b]\}}S(T)] - K \mathbb{E}_t^\mathbb{Q}[\alpha W(T)^3 \mathbf{1}_{\{W(T) \in [a,b]\}}]\Big\}\\
    &=\alpha D(t,T)\Big\{ e^{3\Lambda_W(T) + 3Y(t)e^{-\kappa_W(T-t)} + 3\theta_W(1-e^{-\kappa_W(T-t)}) + \Lambda_S(T)}  \notag \\
    &\quad \times e^{X(t)e^{-\kappa_S(T-t)}  + \theta_S(1-e^{-\kappa_S(T-t)}) + \frac{\sigma_S^2(1-\rho^2)}{4\kappa_S}(1-e^{-2\kappa_S(T-t)}) } \notag\\
    & \quad \times e^{+ \frac{9 \sigma_W^2}{4 \kappa_W}(1-e^{-2\kappa_W (T-t)}) + \frac{\rho^2 \sigma_S^2}{4\kappa_S}(1-e^{-2\kappa_S(T-t)}) + \frac{3\rho\sigma_W\sigma_S}{\kappa_W+\kappa_S}(1-e^{-(\kappa_W+\kappa_S)(T-t)})} \times \\
    & \quad \times \big( \Phi(\bar{b})- \Phi(\bar{a}) \big) + \\
    & \quad - K e^{3\Lambda_W(T) + 3Y(t)e^{-\kappa_W(T-t)} + 3\theta_W(1-e^{-\kappa_W(T-t)}) + \frac{9\sigma_W^2}{4\kappa_W} (1-e^{-2\kappa_W(T-t)}) }  \notag \\
    &\quad  \times \big( \Phi(\tilde{b}) - \Phi(\tilde{a})\big)\Big\} .
\end{align*}
which proves Eq. \eqref{eq:final_ppa_gaussian}.
\end{proof}

\subsection{PPA pricing}\label{app:ppa_gaussian_price}
\begin{proof}
Proof of Proposition \ref{prop:K_gaussian}. \\
We plug the explicit formulas for the conditional expected values obtained in Section \ref{app:ppa_gaussian_val} into Eq. \eqref{eq:price_k_of_wind_ppa} and obtain 
\begin{align}
    K &= \frac{\sum_{j=1}^n \mathbb{E}^\mathbb{Q}[\alpha W(T_j)^3 \mathbf{1}_{\{W(T_j) \in [a,b]\}}S(T_j) D(t,T_j)|\mathcal{F}_t]}{\sum_{j=1}^n \mathbb{E}^\mathbb{Q}[\alpha W(T_j)^3 \mathbf{1}_{\{W(T_j) \in [a,b]\}}D(t,T_j)|\mathcal{F}_t]} \notag \\
    &= \bigg(\sum_{j=1}^n D(t,T_j) e^{3\Lambda_W(T_j) + 3Y(t)e^{-\kappa_W(T_j-t)} + 3\theta_W(1-e^{-\kappa_W(T_j-t)}) + \Lambda_S(T_j)}  \notag \\
    &\quad \times e^{X(t)e^{-\kappa_S(T_j-t)}  + \theta_S(1-e^{-\kappa_S(T_j-t)}) +\frac{\sigma_S^2(1-\rho^2)}{4\kappa_S}(1-e^{-2\kappa_S(T_j-t)}) + \frac{9 \sigma_W^2}{4 \kappa_W}(1-e^{-2\kappa_W (T-t)})} \notag\\
    & \quad \times e^{\frac{\rho^2 \sigma_S^2}{4\kappa_S}(1-e^{-2\kappa_S(T_j-t)}) + \frac{3\rho\sigma_W\sigma_S}{\kappa_W+\kappa_S}(1-e^{-(\kappa_W+\kappa_S)(T_j-t)})} \big( \Phi(\bar{b})- \Phi(\bar{a}) \big) \bigg) \bigg/ \notag \\
    & \quad  \bigg( \sum_{j=1}^n D(t,T_j) e^{3\Lambda_W(T_j) + 3Y(t)e^{-\kappa_W(T_j-t)} + 3\theta_W(1-e^{-\kappa_W(T_j-t)}) +\frac{9\sigma_W^2}{4\kappa_W} (1-e^{-2\kappa_W(T-t)})} \notag \\
    & \quad \times \big( \Phi(\tilde{b}) - \Phi(\tilde{a})\big)\bigg). \notag
\end{align}
\end{proof}

\section{PPA Pricing with the Jump-Diffusion Ornstein-Uhlenbeck model}

\subsection{PPA Valuation}\label{app:ppa_jump_val}
\begin{proof}
Proof of Proposition \ref{prop:val_jump}.\\
Following \cite{risks6020056}, we now assume that the link between wind speed and power production is captured by the following formula 
$$Q(T,W(T)) = \alpha W(T)^3 \mathbf{1}_{\{W(T) \in [a,b]\}}.$$ 
Where $\alpha>0$ is an efficiency production factor. Then, the price of a $UPPA(t, T, W(t), S(t))$ becomes 
\begin{align}
\label{eq:step0_solution_uppa}
    UPPA(t, T, W(t), S(t)) &= 
    D(t,T)\mathbb{E}^\mathbb{Q}_t[\alpha W(T)^3\mathbf{1}_{\{W(T) \in [a,b]\}}(S(T)-K)] \notag \\
    &= \alpha D(t,T)\big\{\mathbb{E}^\mathbb{Q}_t[ W(T)^3\mathbf{1}_{\{W(T) \in [a,b]\}}S(T)] \notag \\
    &\quad - K \mathbb{E}^\mathbb{Q}_t[ W(T)^3\mathbf{1}_{\{W(T) \in [a,b]\}}]\big\} 
\end{align}
If we consider the second conditional expectated value of Eq. \eqref{eq:step0_solution_uppa}, we see that 
\begin{align*}
    \mathbb{E}_t^\mathbb{Q}[ W(T)^3\mathbf{1}_{\{W(T) \in [a,b]\}}] & = e^{3\Lambda_W(T) + 3Y(t)e^{-\kappa_W(T-t)}  + 3\theta_W(1-e^{-\kappa_W(T-t)})} \\
    &\quad  \times \mathbb{E}_t^\mathbb{Q}\bigg[ \exp\bigg\{3\sigma_W\int_t^Te^{-\kappa_W(T-s)} dZ^W(s) +  \\
    & \quad + 3\int_{t}^T e^{-\kappa_W(T-s)}dL^W(s) \bigg\} \\
    &\quad \times \mathbf{1}_{\{W(T) \in [a,b]\}} \bigg] 
\end{align*}
where the expectation cannot unfortunately be calculated explicitly. We therefore need to rely on Monte-Carlo methods or Fourier techniques.
We rewrite the expectation in the following form
\begin{align*}
    \mathbb{E}_t^\mathbb{Q}&\bigg[ \exp\bigg\{3\sigma_W\int_t^Te^{-\kappa_W(T-s)} dZ^W(s) + 3\int_t^T e^{-\kappa_W(T-s)}dL^W(s) \bigg\} \mathbf{1}_{\{W(T) \in [a,b]\}} \bigg] = \\
    &= \mathbb{E}^\mathbb{Q}_t\bigg[e^{3(U + V)}\mathbf{1}_{\{(U+V)\in[\tilde{a}, \tilde{b}] \}}\bigg] 
\end{align*}
where
\begin{align*}
    & \tilde{a} = \log(a) - (\Lambda_W(T) + Y(t) e^{-\kappa_W(T-t)} + \theta_W (1 - e^{-\kappa_W(T-t)} ))
    \\
    & \tilde{b} = \log(b) - (\Lambda_W(T) + Y(t) e^{-\kappa_W(T-t)} + \theta_W (1 - e^{-\kappa_W(T-t)} )) 
    \\
    & U = \sigma_W\int_t^Te^{-\kappa_W(T-s)} dZ^W(s)
    \\
    & V = \int_t^T e^{-\kappa_W(T-s)}dL^W(s).
\end{align*}
We then multiply and divide by $\mathbb{E}^\mathbb{Q}[e^{3(U+V)}]$, in order to isolate the  Radon-Nikodym derivative $\frac{d\tilde{\mathbb{Q}}}{d\mathbb{Q}}\Big|_{T}=\frac{e^{3(U+V)}}{\mathbb{E}^\mathbb{Q}[e^{3(U+V)}]}$, yielding
\begin{align*}
    \mathbb{E}_t^\mathbb{Q}\bigg[ \exp\bigg\{3\sigma_W\int_t^Te^{-\kappa_W(T-s)} dZ^W(s) &+ 3\int_t^T e^{-\kappa_W(T-s)}dL^W(s) \bigg\} \mathbf{1}_{\{W(T) \in [a,b]\}} \bigg] = \\
    &= \mathbb{E}^\mathbb{Q}[e^{3(U+V)}] \mathbb{E}^\mathbb{Q}_t\bigg[ \frac{e^{3(U+V)}}{\mathbb{E}^\mathbb{Q}[e^{3(U+V)}]} \mathbf{1}_{\{(U+V)\in[\tilde{a}, \tilde{b}] \}}\bigg] \\
    &= \mathbb{E}^\mathbb{Q}[e^{3(U+V)}] \tilde{\mathbb{Q}}_t\Big((U+V)\in[\tilde{a}, \tilde{b}]\Big).
\end{align*}
We recall that, by the Gil-Pelaez formula, we have that $P(X<x) = \frac{1}{2} - \frac{1}{\pi} \int_0^{+\infty} \frac{Im\{e^{-i\xi x} \tilde{\varphi}_{U+V}(\xi) \}}{\xi} d\xi$. Then, we can rewrite 
$\tilde{\mathbb{Q}}_{t}\Big((U+V)\in[\tilde{a}, \tilde{b}]\Big) = \bigg(\frac{1}{\pi} \int_0^{+\infty} \frac{Im\{e^{-i\xi \tilde{a}} \tilde{\varphi}_{U+V}(\xi) \}}{\xi} d\xi - \frac{1}{\pi} \int_0^{+\infty} \frac{Im\{e^{-i\xi \tilde{b}} \tilde{\varphi}_{U+V}(\xi) \}}{\xi} d\xi\bigg) $, with $\tilde{\varphi}_{U+V}(\xi) = \mathbb{E}_t^{\tilde{\mathbb{Q}}}[e^{i\xi(U+V)}] $ and 
\begin{align*}
    \mathbb{E}_t^{\tilde{\mathbb{Q}}}[e^{i\xi(U+V)}] &= \mathbb{E}_t^{\mathbb{Q}}\bigg[e^{i\xi(U+V)} \frac{e^{3(U+V)}}{\mathbb{E}^\mathbb{Q}[e^{3(U+V)}]} \bigg] \\
    &= \frac{1}{\mathbb{E}^\mathbb{Q}[e^{3(U+V)}]} \mathbb{E}^\mathbb{Q}[ e^{i(\xi -3i)(U+V)} ] 
    \\
    & = \frac{1}{\mathbb{E}^\mathbb{Q}[e^{3(U+V)}]} \varphi_{U+V}(\xi - 3i) \\
    &= \frac{1}{\mathbb{E}^\mathbb{Q}[e^{3(U+V)}]} \varphi_{U}(\xi - 3i)\varphi_{V}(\xi - 3i)\\
    &= \frac{\varphi_{U}(\xi - 3i)\varphi_{V}(\xi - 3i)}{\varphi_U(-3i)\varphi_V(-3i)},
\end{align*}
where $\varphi_{U}(\xi) = \mathbb{E}_t^{\mathbb{Q}}[e^{i\xi(U)}]$.\\
So, the second expected value of Eq. \eqref{eq:step0_solution_uppa} is finally rewritten as
\begin{align}
    \mathbb{E}_t^\mathbb{Q}[ W(T)^3\mathbf{1}_{\{W(T) \in [a,b]\}}&] =  e^{3\Lambda_W(T) + 3Y(t)e^{-\kappa_W(T-t)} + 3\theta_W(1-e^{-\kappa_W(T-t)})} \notag \\
    &\quad  \times \mathbb{E}_t^\mathbb{Q}\bigg[ \exp\bigg\{3\sigma_W\int_t^Te^{-\kappa_W(T-s)} dZ^W(s) + \notag\\
    &\quad + 3\int_{t}^T e^{-\kappa_W(T-s)}dL^W(s) \bigg\} \notag \\
    & \quad \times \mathbf{1}_{\{W(T) \in [a,b]\}} \bigg] \notag \\
    &= e^{3\Lambda_W(T) + 3Y(t)e^{-\kappa_W(T-t)} + 3\theta_W(1-e^{-\kappa_W(T-t)})} \notag \\
    &\quad \times \mathbb{E}^\mathbb{Q}[e^{3(U+V)}] \tilde{\mathbb{Q}}_{T}\Big((U+V)\in[\tilde{a}, \tilde{b}]\Big) \notag \\
    & = e^{3\Lambda_W(T) + 3Y(t)e^{-\kappa_W(T-t)} + 3\theta_W(1-e^{-\kappa_W(T-t)})} \mathbb{E}^\mathbb{Q}[e^{3(U+V)}] \notag \\
    &\quad \times \bigg(\frac{1}{\pi} \int_0^{+\infty} \frac{Im\{e^{-i\xi \tilde{a}} \tilde{\varphi}_{U+V}(\xi) \}}{\xi} d\xi +\notag \\
    & \quad - \frac{1}{\pi} \int_0^{+\infty} \frac{Im\{e^{-i\xi \tilde{b}} \tilde{\varphi}_{U+V}(\xi) \}}{\xi} d\xi\bigg).\label{eq:second_expected_val_appendixB}
\end{align}

A similar procedure can be repeated for the first expected value of Eq. \eqref{eq:step0_solution_uppa}. We begin by rewriting it as 
\begin{align*}
    \mathbb{E}^\mathbb{Q}[W(T)^3\mathbf{1}_{\{W(T) \in [a,b]\}}S(T)] & = e^{3\Lambda_W(T) + 3Y(t)e^{-\kappa_W(T-t)} + 3\theta_W(1-e^{-\kappa_W(T-t)}) } \\ 
    &\quad \times e^{X(t)e^{-\kappa_S(T-t)} + \theta_S(1-e^{-\kappa_S(T-t)})}  \\
    & \quad \times \mathbb{E}_t^\mathbb{Q}\bigg[ \exp\bigg\{3\sigma_W\int_t^Te^{-\kappa_W(T-s)} dZ^W(s)+ \\
    &\quad + \sigma_S\int_t^T e^{-\kappa_S(T-u)}dZ^S(u) \bigg\}   
    \\ 
    &\quad \times  \exp\bigg\{3\int_t^Te^{-\kappa_W(T-s)} dL^W(s) + \int_t^T e^{-\kappa_S(T-u)}dL^S(u) \bigg\} \\
    & \quad \times \mathbf{1}_{\{W(T) \in [a,b]\}} \bigg].
\end{align*}
We then rewrite the exponential functions inside the expected value as
\begin{align*}
    \mathbb{E}_t^\mathbb{Q}\bigg[ &\exp\bigg\{3\sigma_W\int_t^Te^{-\kappa_W(T-s)} dZ^W(s) + \sigma_S\int_t^T e^{-\kappa_S(T-u)}dZ^S(u) \bigg\}  
    \\ 
    & \quad \times \exp\bigg\{3\int_t^Te^{-\kappa_W(T-s)} dL^W(s) + \int_t^T e^{-\kappa_S(T-u)}dL^S(u) \bigg\} \mathbf{1}_{\{W(T) \in [a,b]\}} \bigg] = \\
    & \quad = \mathbb{E}_t^\mathbb{Q}\bigg[ e^{3(U+V) + A+B} \mathbf{1}_{\{(U+V)\in[\tilde{a}, \tilde{b}] \}}\bigg],
\end{align*}
where
\begin{align*}
    & \tilde{a} = \log(a) - (\Lambda_W(T) + Y(t) e^{-\kappa_W(T-t)} + \theta_W (1 - e^{-\kappa_W(T-t)} ) )
    \\
    & \tilde{b} = \log(b) - (\Lambda_W(T) + Y(t) e^{-\kappa_W(T-t)} + \theta_W (1 - e^{-\kappa_W(T-t)} ))
    \\
    & U = \sigma_W\int_t^Te^{-\kappa_W(T-s)} dZ^W(s)
    \\
    & V = \int_t^T e^{-\kappa_W(T-s)}dL^W(s)
    \\
    & A = \sigma_S\int_t^T e^{-\kappa_S(T-u)}dZ^S(u) 
    \\
    & B = \int_t^T e^{-\kappa_S(T-u)}dL^S(u).
\end{align*}
The expected value then becomes
\begin{align*}
    \mathbb{E}_t^\mathbb{Q}\bigg[e^{3(U+V) + A+B} \mathbf{1}_{\{(U+V)\in[\tilde{a}, \tilde{b}] \}} \bigg] = \mathbb{E}^\mathbb{Q}_t[e^{B}] \mathbb{E}^\mathbb{Q}_t\bigg[e^{3(U+V) + A} \mathbf{1}_{\{(U+V)\in[\tilde{a}, \tilde{b}] \}} \bigg]. 
\end{align*}
We next multiply and divide by $\mathbb{E}^\mathbb{Q}[e^{3(U+V) + A}]$ in order to isolate the Radon-Nikodym derivative
$$\frac{d\bar{\mathbb{Q}}}{d\mathbb{Q}}\bigg|_{T}=\frac{e^{3(U+V) + A}}{\mathbb{E}^\mathbb{Q}[e^{3(U+V) + A}]}$$ 
and obtain
\begin{align*}
    \mathbb{E}^\mathbb{Q}_t[e^{B}] \mathbb{E}^\mathbb{Q}_t\bigg[e^{3(U+V) + A} \mathbf{1}_{\{(U+V)\in[\tilde{a}, \tilde{b}] \}} \bigg] &= \mathbb{E}^\mathbb{Q}_t[e^{B}]\mathbb{E}^\mathbb{Q}[e^{3(U+V) + A}] \\
    & \quad \times \mathbb{E}^\mathbb{Q}_t\bigg[\mathbf{1}_{\{(U+V)\in[\tilde{a}, \tilde{b}] \}}  \frac{e^{3(U+V) + A}}{\mathbb{E}^\mathbb{Q}[e^{3(U+V) + A}]} \bigg] \\
    & = \mathbb{E}^\mathbb{Q}_t[e^{B}]\mathbb{E}^\mathbb{Q}[e^{3(U+V) + A}] \bar{\mathbb{Q}}_{T}\bigg( (U+V)\in[\tilde{a}, \tilde{b}]  \bigg) \\
    &= \mathbb{E}^\mathbb{Q}_t[e^{B}]\mathbb{E}^\mathbb{Q}[e^{3(U+V) + A}] \\
    & \quad \times \bigg(\frac{1}{\pi} \int_0^{+\infty} \frac{Im\{e^{-i\xi \tilde{a}} \Bar{\varphi}_{U+V}(\xi) \}}{\xi} d\xi +\\
    &\quad - \frac{1}{\pi} \int_0^{+\infty} \frac{Im\{e^{-i\xi \tilde{b}} \Bar{\varphi}_{U+V}(\xi) \}}{\xi} d\xi\bigg),
\end{align*}
where $\Bar{\varphi}_{U+V}(\xi) = \mathbb{E}_t^{\Bar{\mathbb{Q}}}[e^{i\xi(U+V)}]$ and
\begin{align*}
    \mathbb{E}_t^{\tilde{\mathbb{Q}}}[e^{i\xi(U+V)}] &= \mathbb{E}_t^{\mathbb{Q}}\bigg[e^{i\xi(U+V)} \frac{e^{3(U+V) + A}}{\mathbb{E}^\mathbb{Q}[e^{3(U+V) + A}]} \bigg] \\
    &= \frac{1}{\mathbb{E}^\mathbb{Q}[e^{3(U+V)+A}]} \mathbb{E}^\mathbb{Q}_{T}[ e^{i(\xi -3i)U + A} ]  \mathbb{E}^\mathbb{Q}_{T}[e^{i(\xi -3i)V}]
    \\
    & = \frac{1}{\mathbb{E}^\mathbb{Q}[e^{3(U+V) + A}]} \varphi_{U,A}(\xi - 3i, -i )\varphi_{V}(\xi - 3i)\\
    & = \frac{\varphi_{U,A}(\xi - 3i, -i )\varphi_{V}(\xi - 3i)}{\varphi_{U,A}(- 3i, -i )\varphi_{V}(- 3i)}.
\end{align*}
So, the first expected value of Eq. \eqref{eq:step0_solution_uppa} is finally rewritten as
\begin{align}
     \mathbb{E}^\mathbb{Q}[W(T)^3\mathbf{1}_{\{W(T) \in [a,b]\}}S(T)] & = e^{3\Lambda_W(T) + 3Y(t)e^{-\kappa_W(T-t)} + 3\theta_W(1-e^{-\kappa_W(T-t)})} \notag\\ 
    &\quad \times e^{X(t)e^{-\kappa_S(T-t)} + \theta_S(1-e^{-\kappa_S(T-t)})} \mathbb{E}^\mathbb{Q}_t[e^{B}]\mathbb{E}^\mathbb{Q}[e^{3(U+V) + A}] \\
    & \quad \times \bigg(\frac{1}{\pi} \int_0^{+\infty} \frac{Im\{e^{-i\xi \tilde{a}} \bar{\varphi}_{U+V}(\xi) \}}{\xi} d\xi + \notag \\
    &\quad - \frac{1}{\pi} \int_0^{+\infty} \frac{Im\{e^{-i\xi \tilde{b}} \bar{\varphi}_{U+V}(\xi) \}}{\xi} d\xi\bigg). \label{eq:first_expected_val_appendixB}
\end{align}

Inserting Eq. \eqref{eq:second_expected_val_appendixB} and \eqref{eq:first_expected_val_appendixB} 
into Eq. \eqref{eq:step0_solution_uppa} yields the final pricing formula for the UPPA
\begin{align}
    UPPA(t, T, W(t), S(t)) &= \alpha D(t,T)\big\{\mathbb{E}^\mathbb{Q}_t[ W(T)^3\mathbf{1}_{\{W(T) \in [a,b]\}}S(T)] \notag \\
    &\quad - K \mathbb{E}^\mathbb{Q}_t[ W(T)^3\mathbf{1}_{\{W(T) \in [a,b]\}}]\big\} \notag \\
    &= \alpha D(t,T)\Bigg\{e^{3\Lambda_W(T) + 3Y(t)e^{-\kappa_W(T-t)} + 3\theta_W(1-e^{-\kappa_W(T-t)})}  \notag\\ 
    &\quad \times e^{X(t)e^{-\kappa_S(T-t)} + \theta_S(1-e^{-\kappa_S(T-t)})} \mathbb{E}^\mathbb{Q}_t[e^{B}]\mathbb{E}^\mathbb{Q}[e^{3(U+V) + A}] \notag \\
    & \quad \times \bigg(\frac{1}{\pi} \int_0^{+\infty} \frac{Im\{e^{-i\xi \tilde{a}} \bar{\varphi}_{U+V}(\xi) \}}{\xi} d\xi + \notag \\
    &\quad - \frac{1}{\pi} \int_0^{+\infty} \frac{Im\{e^{-i\xi \tilde{b}} \bar{\varphi}_{U+V}(\xi) \}}{\xi} d\xi\bigg) + \notag \\
    & \quad - K \Bigg( e^{3\Lambda_W(T) + 3Y(t)e^{-\kappa_W(T-t)} + 3\theta_W(1-e^{-\kappa_W(T-t)})} \mathbb{E}^\mathbb{Q}[e^{3(U+V)}] \notag \\
    &\quad \times \bigg(\frac{1}{\pi} \int_0^{+\infty} \frac{Im\{e^{-i\xi \tilde{a}} \tilde{\varphi}_{U+V}(\xi) \}}{\xi} d\xi + \notag \\
    & \quad - \frac{1}{\pi} \int_0^{+\infty} \frac{Im\{e^{-i\xi \tilde{b}} \tilde{\varphi}_{U+V}(\xi) \}}{\xi} d\xi\bigg) \Bigg) \Bigg\}, \notag
\end{align}
where 
$$\bar{\varphi}_{U+V}(\xi) =  \frac{\varphi_{U,A}(\xi - 3i, -i )\varphi_{V}(\xi - 3i)}{\varphi_{U,A}(- 3i, -i )\varphi_{V}(- 3i)}$$
and
$$\tilde{\varphi}_{U+V}(\xi) = \frac{\varphi_{U}(\xi - 3i)\varphi_{V}(\xi - 3i)}{\varphi_U(-3i)\varphi_V(-3i)},$$
hence proving Eq. \eqref{eq:final_ppa_jump}.
\end{proof}

\subsection{PPA Pricing}\label{app:ppa_jump_price}
\begin{proof}
Proof of Proposition \ref{prop:K_jump}.\\
We plug the explicit formulas for the conditional expected values obtained in Section \ref{app:ppa_jump_val} into Eq. \eqref{eq:price_k_of_wind_ppa} and obtain 
\begin{align}
    K &= \frac{\sum_{j=1}^n \mathbb{E}^\mathbb{Q}[\alpha W(T_j)^3 \mathbf{1}_{\{W(T_j) \in [a,b]\}}S(T_j) D(t,T_j)|\mathcal{F}_t]}{\sum_{j=1}^n \mathbb{E}^\mathbb{Q}[\alpha W(T_j)^3 \mathbf{1}_{\{W(T_j) \in [a,b]\}}D(t,T_j)|\mathcal{F}_t]} \notag \\
    &=\bigg( \sum_{j=1}^n D(t, T_j)e^{3\Lambda_W(T_j) + 3Y(t)e^{-\kappa_W(T_j-t)} + 3\theta_W(1-e^{-\kappa_W(T_j-t)}) + X(t)e^{-\kappa_S(T_j-t)}}  \notag\\ 
    &\quad \times e^{\theta_S(1-e^{-\kappa_S(T_j-t)})} \mathbb{E}^\mathbb{Q}_t[e^{B}]\mathbb{E}^\mathbb{Q}[e^{3(U+V) + A}] \big(\frac{1}{\pi} \int_0^{+\infty} \frac{Im\{e^{-i\xi \tilde{a}} \bar{\varphi}_{U+V}(\xi) \}}{\xi} d\xi \notag \\
    &\quad - \frac{1}{\pi} \int_0^{+\infty} \frac{Im\{e^{-i\xi \tilde{b}} \bar{\varphi}_{U+V}(\xi) \}}{\xi} d\xi\big) \bigg) \bigg/ \notag \\
    &\quad \Bigg( \sum_{j=1}^n D(t, T_j) e^{3\Lambda_W(T_j) + 3Y(t)e^{-\kappa_W(T_j-t)} + 3\theta_W(1-e^{-\kappa_W(T_j-t)})} \mathbb{E}^\mathbb{Q}[e^{3(U+V)}] \notag \\
    &\quad \times \bigg(\frac{1}{\pi} \int_0^{+\infty} \frac{Im\{e^{-i\xi \tilde{a}} \tilde{\varphi}_{U+V}(\xi) \}}{\xi} d\xi - \frac{1}{\pi} \int_0^{+\infty} \frac{Im\{e^{-i\xi \tilde{b}} \tilde{\varphi}_{U+V}(\xi) \}}{\xi} d\xi\bigg) \Bigg)
\end{align}
\end{proof}

\end{appendix}

\end{document}